\documentclass[11pt]{article}

\usepackage[margin=1in]{geometry}
\usepackage{amsmath,amssymb,amsfonts,amsthm,mathtools}
\usepackage{graphicx}
\usepackage{booktabs}
\usepackage{float}
\usepackage{hyperref}
\usepackage{enumitem}
\usepackage{algorithm}
\usepackage{algpseudocode}
\usepackage{algorithmicx}
\usepackage{setspace}
\usepackage{natbib}
\usepackage{caption}
\usepackage{subcaption}
\usepackage{xcolor}

\newtheorem{definition}{Definition}[section]
\newtheorem{assumption}{Assumption}[section]
\newtheorem{lemma}{Lemma}[section]
\newtheorem{proposition}{Proposition}[section]

\newtheorem{remark}{Remark}[section]

\begin{document}
 \begin{center}
\textbf{\Large Pricing Excess-of-Loss Reinsurance and CAT Bonds under Climate Uncertainty:\\
A Cox Process Framework with Temperature-Dependent Stochastic Intensity}

 N. Karimi $^{a}$,  F. Shokrollahi\footnote{Corresponding Author, \\ Email: foad.shokrollahi@uwasa.fi}
\end{center}

\begin{center}

\emph{\footnotesize $^{a}$ Department of Applied Mathematics, Faculty of Mathematics and Computer Science} \\
\emph{\footnotesize Amirkabir University of Technology, No. 424, Hafez Ave.,15914, Tehran, Iran}\\
\emph{\footnotesize$^{1}$  Department of Mathematics and Statistics, University of Vaasa, Vaasa, Finland}\\
\end{center}

\begin{abstract}
This paper develops a climate-aware pricing framework for excess-of-loss (XL) reinsurance contracts and catastrophe (CAT) bonds under non-stationary catastrophe risk. Catastrophe arrivals are modeled as a Cox process whose stochastic intensity depends exponentially on a temperature-related climate index. To represent climate dynamics, the index is modeled as a mean-reverting Ornstein--Uhlenbeck process around a time-dependent warming trend.

Within this setting, aggregate losses follow a compound Cox structure with lognormal severities. Pricing is performed under a reduced-form risk-adjusted measure, which provides a tractable valuation approach for XL reinsurance layers and binary zero-coupon CAT bond payoffs in an incomplete market setting. Because catastrophe losses are not dynamically replicable, the framework emphasizes scenario-based valuation rather than model-independent no-arbitrage bounds.

A Monte Carlo valuation scheme is implemented to quantify the economic implications of climate-dependent catastrophe intensity. The numerical results show that climate dependence materially changes the loss-generation mechanism and affects the valuation of catastrophe-linked contracts. In the baseline calibration, the climate-aware model increases the excess-of-loss reinsurance premium and lowers the CAT bond price relative to the stationary benchmark. Furthermore, our analysis of the 99.5\% Tail Value-at-Risk (TVaR) indicates that stationary benchmarks may underestimate economic capital requirements by approximately 13.7\% compared to the climate-aware framework, highlighting the potential regulatory relevance of the proposed model. This finding highlights that benchmark design is critical for interpreting climate-pricing effects.

The framework is methodological in emphasis. It provides a transparent stochastic architecture that can be extended through richer empirical calibration, alternative risk-adjustment mechanisms, market-implied parameter estimation, and climate-dependent severity effects.
\end{abstract}

\noindent\textbf{Keywords:} climate risk; catastrophe bonds; excess-of-loss reinsurance; Cox process; stochastic intensity; compound Poisson process; non-stationarity; Monte Carlo pricing.

\noindent\textbf{JEL Classification:} G13, G22, C63, Q54, C15.

\section{Introduction}

The escalating frequency and severity of extreme weather events have positioned climate change as a systemic threat to the global economy and the insurance sector. According to the Intergovernmental Panel on Climate Change (IPCC), human-induced climate change is already affecting many weather and climate extremes across the globe \citep{ipcc2021}. For the insurance and reinsurance industries, which act as primary shock absorbers for financial losses caused by natural catastrophes, this shift implies that historical loss data may no longer provide a reliable guide to future risks \citep{SwissRe2023}. Traditional actuarial frameworks often rely on stationarity, treating the frequency and severity of catastrophic events as stable over time. However, as argued by \citet{milly2008stationarity}, this assumption becomes increasingly questionable in the presence of persistent environmental change.

Despite the growing consensus on physical climate risks, their integration into financial pricing models, particularly for insurance-linked securities (ILS) such as catastrophe (CAT) bonds and reinsurance treaties, remains limited. Existing approaches often treat climate change as a static parameter shift or rely on scenario analysis without embedding stochastic climate dynamics directly into the pricing engine \citep{Botzen2021}. This disconnect creates a gap between the physical models used in climate science and the financial models used by insurers, reinsurers, and capital-market participants. In particular, a homogeneous Poisson process with constant intensity cannot capture the compounding effect of gradual warming on catastrophe occurrence probabilities.

This paper bridges this gap by proposing a tractable stochastic framework that explicitly incorporates climate dynamics into catastrophe arrival intensity. Event arrivals are modeled using a Cox, or doubly stochastic Poisson, process whose intensity is driven by a temperature-related climate index. The climate index is modeled as a Gaussian mean-reverting process around a deterministic time-dependent mean, thereby allowing the process to capture fluctuations around a warming trend while preserving analytical tractability. The catastrophe intensity is specified as an exponential function of the climate index, capturing the nonlinear relationship between climate conditions and extreme-event frequency \citep{katz2013statistical}.

The contribution of this paper is fourfold. First, we introduce a climate-dependent intensity model that relaxes the stationarity assumption while retaining a tractable probabilistic structure. By deriving analytical expressions for the moments of the integrated intensity, we show how climate volatility and the warming trend jointly affect expected event counts. Second, we address market incompleteness by formulating valuation under a reduced-form risk-adjusted pricing measure and by reporting scenario-based valuation ranges rather than claiming unique arbitrage-free prices. Third, we provide empirical grounding by calibrating the climate dynamics to historical global temperature anomaly data from NASA GISTEMP. Importantly, the climate dynamics parameters are separated from catastrophe-frequency parameters, which require catastrophe event data, catastrophe-model output, market quotes, or scenario assumptions. Fourth, we analyze benchmark design by comparing the non-stationary climate-dependent model with a stationary benchmark matched to the same initial catastrophe intensity level, thereby isolating the effect of warming trends and climate volatility.

Our numerical analysis, based on Monte Carlo simulation with exact discretization of the climate process, quantifies these effects for excess-of-loss reinsurance and stylized CAT bond payoffs. The results show that while short-term pricing impacts may be moderate under baseline calibrations, the divergence between stationary and non-stationary models becomes more pronounced for longer horizons, high-sensitivity scenarios, and tail-risk measures such as TVaR. These findings have implications for regulatory capital assessment, multi-year reinsurance design, and climate-aware insurance risk management.

The rest of the paper is organized as follows. Section~\ref{sec:lit} reviews the relevant literature on collective risk models, Cox processes, catastrophe-linked securities, climate risk modeling, and incomplete-market valuation. Section~\ref{sec:prob} introduces the probabilistic framework, including the climate index dynamics, the climate-dependent Cox process, and the analytical moments of the integrated intensity. Section~\ref{sec:risk} develops the reduced-form risk-adjusted pricing framework for incomplete markets, defines payoff structures for XL reinsurance and CAT bonds, and introduces scenario-based valuation ranges. Section~\ref{sec:mont} outlines the Monte Carlo simulation methodology used for numerical valuation. Section~\ref{sec:numerical} presents the numerical and empirical analysis, including calibration, distributional results, convergence checks, sensitivity testing, long-term pricing dynamics, and tail-risk assessment. Finally, Section~\ref{sec:con} concludes the paper and discusses implications for climate-aware insurance risk management.

\section{Literature Review}
\label{sec:lit}
The paper is related to several strands of literature, connecting actuarial science, financial mathematics, climate science, and risk management.

The actuarial literature on collective risk models and compound Poisson processes provides the classical foundation for aggregate claims modeling. Standard references include \citet{asmussen2000ruin}, \citet{klugman2012loss}, and \citet{grandell1997mixed}. In the classical model, claim arrivals are often assumed to follow a homogeneous Poisson process. While analytically convenient, this assumption is restrictive when exposures or environmental conditions evolve over time. To address unobserved heterogeneity, mixed Poisson and Cox process models have been developed, which allow flexible intensity modulation. General treatments of point processes and stochastic intensities can be found in \citet{daley2003introduction} and \citet{bremaud1981point}, providing the theoretical basis for the stochastic intensity framework used in this study.

The paper builds on the theory of Cox processes and stochastic intensity models. Cox processes were introduced by \citet{cox1955some} and have since become central in credit risk, insurance, reliability, and point process modeling. General treatments of point processes and stochastic intensities can be found in \citet{daley2003introduction}, \citet{daley2008introduction}, and \citet{bremaud1981point}. In insurance, mixed Poisson and Cox process models provide flexible alternatives to homogeneous Poisson assumptions, especially when unobserved heterogeneity or time-varying risk factors matter.

Third, the paper is connected to catastrophe risk and insurance-linked securities pricing. CAT bonds transfer catastrophe risk from insurers and reinsurers to capital markets. Early and influential contributions include \citet{cummins2008cat}, \citet{doherty1997financial}, \citet{froot2001market}, and \citet{lane2000pricing}. These studies emphasize that CAT bond pricing must account for loss triggers, basis risk, tail dependence, illiquidity, and risk premia.

Fourth, the paper relates to climate-change economics and climate risk modeling. The non-stationarity of climate-related hazards has been widely discussed in climate science and risk management; see, for example, \citet{ipcc2021}, \citet{katz2013statistical}, and \citet{milly2008stationarity}. Recent comprehensive reviews by \citet{ingels2024art} and \citet{botzen2024handbook} highlight the growing sophistication of climate risk modeling in insurance, emphasizing the need to move beyond static assumptions. \citet{courbage2022extreme} further underscore the critical role of insurance mechanisms in managing extreme events under changing climate conditions. These studies motivate the use of non-stationary hazard models in insurance pricing, where reliance on historical event sets and stationary frequency assumptions may understate future risk.

In addition to the foundational works cited above, recent empirical and theoretical developments have expanded the understanding of climate risk pricing and market dynamics. In the CAT bond market, studies such as \citet{gatzert2021climate} document the evolving role of catastrophe-linked securities as a diversification tool, while also highlighting liquidity premiums that affect pricing models in incomplete markets.

Methodologically, non-stationary point processes, such as Hawkes processes, have increasingly been applied to model temporal clustering in extreme events (e.g., wildfires and hurricanes) \citet{Chavez-Demoulin2022}. These models challenge the independence assumption of standard Cox processes and suggest that ignoring clustering can lead to tail-risk underestimation. Although our framework retains the Cox-process assumption for tractability, we acknowledge these developments as an important avenue for future extensions.

Furthermore, the intersection of physical climate risk and financial stability has become a central theme. Research on "stranded assets" \citet{battiston2017climate} and "climate value at risk" provides mechanisms for assessing the impact of transition risks. \citet{basten2024who} examine how financial institutions monitor and incorporate climate risk, linking physical catastrophe risks to broader financial stability metrics, a concern also central to our regulatory capital assessment. Our scenario-based valuation framework aligns with these regulatory stress-testing approaches (e.g., by adverse scenario selection), ensuring our pricing reflects both physical and financial dimensions of climate risk.

The landscape of climate risk insurance is rapidly evolving, as documented by recent bibliometric analyses \citet{horani2025climate}. Research by \citet{singh2024weather} explores the viability of weather index insurance as an alternative mechanism for transferring climate risk, complementing capital market instruments like CAT bonds. While our paper focuses on CAT bonds and reinsurance, the broader discourse on index insurance and agricultural strategies \citet{madaki2023agricultural} underscores the diverse adaptation pathways being explored globally. \citet{garrido2024climate} provide further insights into the multifaceted impact of climate risk on the insurance sector, reinforcing the necessity of robust, dynamic pricing models.

Fifth, the paper is related to incomplete-market valuation and actuarial premium principles. Since catastrophe risks are only partially hedgeable, there is generally no unique equivalent martingale measure. Pricing therefore often relies on risk-adjusted measures, distortion principles, Wang transforms, equilibrium-based pricing, or reduced-form pricing kernels; see \citet{wang2000class}, \citet{gerber1994option}, and \citet{buhlmann1980economic}. In the present setting, lognormal severities make a positive loss-level Esscher transform problematic because the moment generating function of a lognormal random variable is infinite for positive arguments. This motivates the use of a reduced-form risk-adjusted pricing measure acting on the climate and loss-generation mechanism.

Overall, the literature supports the methodological choices made in this paper. The combination of stochastic climate dynamics (Section~\ref{sec:prob}) and reduced-form pricing (Section~\ref{sec:risk}) provides a comprehensive framework for pricing climate-exposed insurance-linked instruments, bridging the gap between environmental science and financial risk assessment.

\section{Probabilistic Model}\label{sec:prob}

\subsection{Filtered Probability Space}
Let $(\Omega,\mathcal{F},\mathbb{F},\mathbb{P})$ be a complete filtered probability space, where $\mathbb{F}=\{\mathcal{F}_t\}_{t\geq 0}$ satisfies the usual conditions. The probability measure $\mathbb{P}$ represents the physical measure. Pricing will be performed under a risk-adjusted measure $\mathbb{Q}$ introduced later.

\subsection{Climate Index Dynamics}

\begin{definition}[Climate index]
The process $X=\{X_t\}_{t\in[0,T]}$ denotes a temperature-related climate index, such as a global or regional temperature anomaly. Positive values represent warmer-than-baseline states.
\end{definition}

\begin{assumption}[Mean-reverting climate dynamics around a warming trend]
Under the physical measure $\mathbb{P}$, the climate index follows the Gaussian mean-reverting dynamics
\begin{equation}
    dX_t=\kappa\big(\theta(t)-X_t\big)dt+\sigma_X dW_t,
    \qquad X_0=x_0,
    \label{eq:climateP}
\end{equation}
where $\kappa>0$ is the speed of mean reversion, $\sigma_X>0$ is the volatility of the climate index, and $\theta(t)$ is a deterministic time-dependent mean level. In the baseline specification,
\begin{equation}
    \theta(t)=\theta_0+\theta_1 t,
    \label{eq:lineartrend}
\end{equation}
where $\theta_1$ captures the deterministic warming trend.
\end{assumption}

\begin{remark}
When $\theta(t)$ is constant, equation \eqref{eq:climateP} reduces to the standard Ornstein--Uhlenbeck process. When $\theta(t)$ is time-dependent, the process remains Gaussian and mean-reverting, but it is no longer stationary. Hence, in this paper the term ``OU-type process'' refers to a Gaussian mean-reverting process around a deterministic moving mean.
\end{remark}

The explicit solution of \eqref{eq:climateP} is
\begin{equation}
    X_t
    =
    x_0 e^{-\kappa t}
    +
    \kappa \int_0^t e^{-\kappa(t-s)}\theta(s)\,ds
    +
    \sigma_X \int_0^t e^{-\kappa(t-s)}\,dW_s.
    \label{eq:Xsolution}
\end{equation}

For the linear trend specification $\theta(t)=\theta_0+\theta_1 t$, the mean is
\begin{equation}
    m_X(t):=\mathbb{E}^{\mathbb{P}}[X_t]
    =
    x_0 e^{-\kappa t}
    +
    \theta_0(1-e^{-\kappa t})
    +
    \theta_1
    \left[
        t-\frac{1-e^{-\kappa t}}{\kappa}
    \right],
    \label{eq:meanX}
\end{equation}
and the variance is
\begin{equation}
    v_X(t):=\operatorname{Var}^{\mathbb{P}}(X_t)
    =
    \frac{\sigma_X^2}{2\kappa}
    \left(1-e^{-2\kappa t}\right).
    \label{eq:varX}
\end{equation}

The covariance function is
\begin{equation}
    c_X(s,t):=\operatorname{Cov}^{\mathbb{P}}(X_s,X_t)
    =
    \frac{\sigma_X^2}{2\kappa}
    \left(
        e^{-\kappa |t-s|}
        -
        e^{-\kappa(t+s)}
    \right),
    \qquad s,t\geq 0,
    \label{eq:covX}
\end{equation}
when the initial value $X_0=x_0$ is deterministic.

\subsection{Climate-Dependent Cox Process}

\begin{definition}[Cox process]
A counting process $N=\{N_t\}_{t\in[0,T]}$ is a Cox process with stochastic intensity $\lambda=\{\lambda_t\}_{t\in[0,T]}$ if, conditional on the path of $\lambda$, $N$ is an inhomogeneous Poisson process satisfying
\begin{equation}
    N_T \mid \mathcal{F}_T^\lambda \sim \text{Poisson}\left(\int_0^T \lambda_s ds\right).
\end{equation}
\end{definition}

\begin{assumption}[Temperature-dependent stochastic intensity]
The catastrophe arrival intensity is
\begin{equation}
    \lambda_t = \lambda_0 \exp(\beta X_t), \qquad \lambda_0>0,\quad \beta\geq 0.
    \label{eq:intensity}
\end{equation}
The exponential specification implies that catastrophe arrival intensity responds nonlinearly to changes in the climate index. Even moderate increases in the climate state variable may generate disproportionately large increases in expected event frequency.

Figure~\ref{fig:intensity_sensitivity} illustrates the functional relationship between the climate index and the catastrophe arrival intensity. The solid lines represent different climate-sensitivity scenarios ($\beta$). The shaded area highlights the substantial uncertainty in catastrophe intensity arising from parameter misspecification, demonstrating the convex amplification of risk as the climate index increases.
\begin{figure}[H]
\centering
\includegraphics[width=0.85\textwidth]{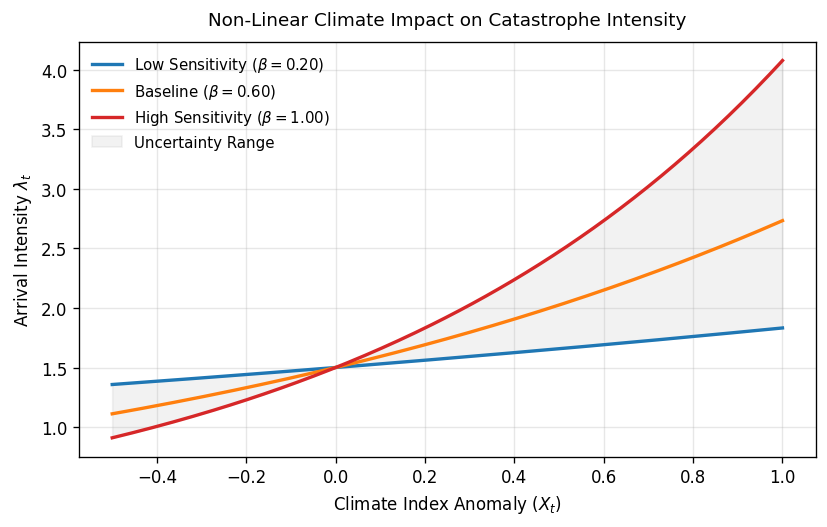}
\caption{Non-linear climate impact on catastrophe intensity across sensitivity scenarios. The shaded region represents the range of possible intensities given uncertainty in the climate-sensitivity parameter $\beta$.}
\label{fig:intensity_sensitivity}
\end{figure}

The integrated intensity is
\begin{equation}
    \Lambda_T = \int_0^T \lambda_0 \exp(\beta X_t)dt.
    \label{eq:Lambda}
\end{equation}
Conditional on the climate path,
\begin{equation}
    N_T\mid \mathcal{F}_T^X \sim \text{Poisson}(\Lambda_T).
    \label{eq:conditionalPoisson}
\end{equation}
\end{assumption}

\subsection{Severity and Aggregate Loss}

\begin{assumption}[Claim severities]
Individual claim severities $\{Y_i\}_{i\geq1}$ are independent and identically distributed, independent of $N$ conditional on the climate path, and follow a lognormal distribution:
\begin{equation}
    Y_i\sim \mathrm{Lognormal}(\mu_Y,\sigma_Y^2).
\end{equation}
Thus,
\begin{equation}
    m_1:=\mathbb{E}[Y_i]=\exp\left(\mu_Y+\frac12\sigma_Y^2\right),
\end{equation}
and
\begin{equation}
    m_2:=\mathbb{E}[Y_i^2]=\exp\left(2\mu_Y+2\sigma_Y^2\right).
\end{equation}
\end{assumption}

A natural extension would allow severity to depend on climate conditions as well. For example, one may specify
\begin{equation}
    \log Y_i \mid X_{T_i}\sim \mathcal{N}(\mu_0+\gamma X_{T_i},\sigma_Y^2),
\end{equation}
so that adverse climate states affect both claim frequency and claim size. In the present paper, however, we retain a climate-independent severity specification in order to isolate the frequency channel and keep the proof-of-concept framework analytically transparent.
\subsection{Analytical Moments of the Integrated Intensity}
Recall that the integrated intensity $\Lambda_T$ is defined in \eqref{eq:Lambda}. Since $X_t$ is Gaussian for each $t$, the exponential moment
$\mathbb{E}^{\mathbb{P}}[\exp(\beta X_t)]$ exists for all finite $\beta$ and finite $t$. Therefore, by Fubini's theorem,
\begin{equation}
    \mathbb{E}^{\mathbb{P}}[\Lambda_T]
    =
    \lambda_0
    \int_0^T
    \mathbb{E}^{\mathbb{P}}\left[
        \exp(\beta X_t)
    \right]dt.
\end{equation}

\begin{proposition}[Expected integrated intensity]
Let $X_t$ follow \eqref{eq:climateP} with deterministic initial value $X_0=x_0$ and deterministic mean function $\theta(t)$. Then
\begin{equation}
    \mathbb{E}^{\mathbb{P}}[\Lambda_T]
    =
    \lambda_0
    \int_0^T
    \exp\left(
        \beta m_X(t)
        +
        \frac{1}{2}\beta^2 v_X(t)
    \right)dt,
    \label{eq:ELambdaCorrect}
\end{equation}
where
\begin{equation}
    m_X(t)=\mathbb{E}^{\mathbb{P}}[X_t],
    \qquad
    v_X(t)=\operatorname{Var}^{\mathbb{P}}(X_t).
\end{equation}
For the linear trend $\theta(t)=\theta_0+\theta_1t$, $m_X(t)$ and $v_X(t)$ are given by \eqref{eq:meanX} and \eqref{eq:varX}.
\end{proposition}

\begin{proof}
Since $X_t$ is Gaussian with mean $m_X(t)$ and variance $v_X(t)$,
\begin{equation}
    \mathbb{E}^{\mathbb{P}}\left[e^{\beta X_t}\right]
    =
    \exp\left(
        \beta m_X(t)
        +
        \frac{1}{2}\beta^2 v_X(t)
    \right).
\end{equation}
Substituting this expression into the definition of $\Lambda_T$ and applying Fubini's theorem yields \eqref{eq:ELambdaCorrect}.
\end{proof}

\begin{remark}[Convexity adjustment]
The term $\frac{1}{2}\beta^2 v_X(t)$ in \eqref{eq:ELambdaCorrect} is a convexity adjustment induced by the exponential intensity specification. It implies that, even for a fixed mean climate trajectory, higher climate volatility increases expected catastrophe frequency.
\end{remark}
\begin{proposition}[Variance of the integrated intensity]
Let $X_t$ follow \eqref{eq:climateP} with deterministic initial value $X_0=x_0$. Then the variance of the integrated intensity $\Lambda_T$ is
\begin{equation}
    \operatorname{Var}^{\mathbb{P}}(\Lambda_T)
    =
    \lambda_0^2
    \int_0^T\int_0^T
    \left[
        \exp\left(
            \beta m_X(s)
            +
            \beta m_X(t)
            +
            \frac{1}{2}\beta^2 v_X(s)
            +
            \frac{1}{2}\beta^2 v_X(t)
            +
            \beta^2 c_X(s,t)
        \right)
    \right]ds\,dt
    -
    \left(
        \mathbb{E}^{\mathbb{P}}[\Lambda_T]
    \right)^2,
    \label{eq:VarLambdaCorrect}
\end{equation}
where
\[
    m_X(t)=\mathbb{E}^{\mathbb{P}}[X_t],
    \qquad
    v_X(t)=\operatorname{Var}^{\mathbb{P}}(X_t),
    \qquad
    c_X(s,t)=\operatorname{Cov}^{\mathbb{P}}(X_s,X_t).
\]
For the linear trend specification, these functions are given by
\eqref{eq:meanX}, \eqref{eq:varX}, and \eqref{eq:covX}.
\end{proposition}

\begin{proof}
By definition,
\begin{equation}
    \mathbb{E}^{\mathbb{P}}[\Lambda_T^2]
    =
    \lambda_0^2
    \int_0^T\int_0^T
    \mathbb{E}^{\mathbb{P}}
    \left[
        \exp\left(\beta X_s+\beta X_t\right)
    \right]ds\,dt.
\end{equation}
Since $(X_s,X_t)$ is jointly Gaussian,
\begin{equation}
    \mathbb{E}^{\mathbb{P}}
    \left[
        \exp\left(\beta X_s+\beta X_t\right)
    \right]
    =
    \exp\left(
        \beta m_X(s)
        +
        \beta m_X(t)
        +
        \frac{1}{2}\beta^2 v_X(s)
        +
        \frac{1}{2}\beta^2 v_X(t)
        +
        \beta^2 c_X(s,t)
    \right).
\end{equation}
Therefore,
\begin{equation}
    \operatorname{Var}^{\mathbb{P}}(\Lambda_T)
    =
    \mathbb{E}^{\mathbb{P}}[\Lambda_T^2]
    -
    \left(
        \mathbb{E}^{\mathbb{P}}[\Lambda_T]
    \right)^2,
\end{equation}
which gives \eqref{eq:VarLambdaCorrect}.
\end{proof}

\begin{definition}[Aggregate loss]
The aggregate catastrophe loss at maturity $T$ is
\begin{equation}
    S_T = \sum_{i=1}^{N_T}Y_i,
\end{equation}
with the convention that $S_T=0$ when $N_T=0$.
\end{definition}

\begin{lemma}[Conditional compound Poisson transform]
Conditional on $\mathcal{F}_T^X$, the Laplace transform of $S_T$ is
\begin{equation}
    \mathbb{E}\left[e^{-uS_T}\mid\mathcal{F}_T^X\right]
    =\exp\left\{\Lambda_T\left(\mathcal{L}_Y(u)-1\right)\right\},\qquad u\geq0,
\end{equation}
where $\mathcal{L}_Y(u)=\mathbb{E}[e^{-uY}]$ is the severity Laplace transform.
\end{lemma}

\begin{proof}
Given $\mathcal{F}_T^X$, $N_T$ is Poisson with mean $\Lambda_T$. Therefore,
\begin{equation}\begin{aligned}
\mathbb{E}[e^{-uS_T}\mid\mathcal{F}_T^X]
&=\sum_{n=0}^{\infty}\mathbb{E}\left[e^{-u\sum_{i=1}^nY_i}\right]
\frac{e^{-\Lambda_T}\Lambda_T^n}{n!}\\
&=\sum_{n=0}^{\infty}\left(\mathcal{L}_Y(u)\right)^n
\frac{e^{-\Lambda_T}\Lambda_T^n}{n!}\\
&=e^{-\Lambda_T}\exp\{\Lambda_T\mathcal{L}_Y(u)\}
=\exp\{\Lambda_T(\mathcal{L}_Y(u)-1)\}.
\end{aligned}\end{equation}
\end{proof}

\begin{proposition}[First moment of aggregate loss]
The conditional expectation of $S_T$ is
\begin{equation}
    \mathbb{E}[S_T\mid\mathcal{F}_T^X]=\Lambda_T m_1,
\end{equation}
and hence
\begin{equation}
    \mathbb{E}[S_T]=m_1\mathbb{E}[\Lambda_T].
\end{equation}
\end{proposition}

\begin{proof}
The first conditional moment formula is standard for compound Poisson sums and the unconditional result follows from the tower property.
\end{proof}
\subsection{Extension: Climate-Dependent Severity}
\label{subsec:severity_extension}

While the baseline model assumes independence between severity and climate, empirical evidence suggests that extreme climate states may exacerbate the economic magnitude of catastrophes. We extend the model to allow for severity dependence.

Let the log-severity depend linearly on the climate index at the time of the event:
\begin{equation} \label{eq:dep_severity}
\log Y_i \mid X_{T_i} \sim \mathcal{N}(\mu_Y + \gamma X_{T_i}, \sigma_Y^2),
\end{equation}
where $\gamma$ captures the climate elasticity of claim severity. If $\gamma > 0$, warmer climates lead to proportionally larger losses.

Under this specification, the conditional expectation of aggregate losses becomes:
\begin{equation} \label{eq:cond_exp_sev}
\mathbb{E}[S_T \mid \mathcal{F}_T^X] = \int_0^T \lambda_0 e^{\beta X_s} \cdot \exp\left( \mu_Y + \gamma X_s + \frac{1}{2}\sigma_Y^2 \right) ds = e^{\mu_Y + \frac{1}{2}\sigma_Y^2} \lambda_0 \int_0^T \exp\left( (\beta + \gamma) X_s \right) ds.
\end{equation}

\begin{proposition}[Effective climate sensitivity under combined frequency--severity dependence]
Let the severity follow \eqref{eq:dep_severity} with $\gamma \geq 0$. Then:
\begin{enumerate}[label=(\roman*)]
  \item The conditional expectation of aggregate losses satisfies
  \[
  \mathbb{E}[S_T \mid \mathcal{F}^X_T] = e^{\mu_Y + \frac{1}{2}\sigma_Y^2} \cdot
  \lambda_0 \int_0^T \exp\bigl((\beta + \gamma)X_s\bigr)\,ds,
  \]
  so the effective climate-sensitivity parameter governing aggregate expected losses is
  $\beta_{\text{eff}} = \beta + \gamma$.
  \item The unconditional expected aggregate loss is
  \[
  \mathbb{E}[S_T] = e^{\mu_Y + \frac{1}{2}\sigma_Y^2} \cdot \lambda_0
  \int_0^T \exp\Bigl((\beta + \gamma)m_X(t) + \tfrac{1}{2}(\beta+\gamma)^2 v_X(t)\Bigr)\,dt,
  \]
  where $m_X(t)$ and $v_X(t)$ are given by Equations~(4) and~(5).
  \item The combined effect of frequency and severity climate dependence is super-additive:
  the pricing effect of joint dependence ($\gamma > 0$) exceeds the sum of the individual
  frequency effect ($\beta > 0$, $\gamma = 0$) and severity effect ($\beta = 0$, $\gamma > 0$)
  taken separately, due to the convexity of the exponential function.
\end{enumerate}
\end{proposition}

\begin{proof}
We prove each part sequentially.

\noindent \textbf{Part (i):}
Recall the definition of aggregate losses $S_T = \sum_{i=1}^{N_T} Y_i$. Conditional on the filtration $\mathcal{F}^X_T$ generated by the climate process, the catastrophe intensity is deterministic. Using the law of total expectation for compound Cox processes and the conditional independence of severities $Y_i$, we have:
\begin{align*}
\mathbb{E}[S_T \mid \mathcal{F}^X_T] &= \mathbb{E}\left[ \sum_{i=1}^{N_T} Y_i \,\Bigg|\, \mathcal{F}^X_T \right] \\
&= \int_0^T \mathbb{E}[Y \mid X_s] \lambda_s \, ds.
\end{align*}
Given the log-normal severity assumption, $\mathbb{E}[Y \mid X_s] = \exp\left(\mu_Y + \gamma X_s + \frac{1}{2}\sigma_Y^2\right)$. The intensity is $\lambda_s = \lambda_0 \exp(\beta X_s)$. Substituting these into the integral yields:
\[
\mathbb{E}[S_T \mid \mathcal{F}^X_T] = e^{\mu_Y + \frac{1}{2}\sigma_Y^2} \lambda_0 \int_0^T \exp\left((\beta + \gamma) X_s\right) ds.
\]
This demonstrates that the combined climate sensitivity parameter governing the path-dependent expectation is $\beta_{\text{eff}} = \beta + \gamma$.

\noindent \textbf{Part (ii):}
To find the unconditional expectation, we take the expectation of the result in part (i) with respect to the climate process $X_t$:
\[
\mathbb{E}[S_T] = e^{\mu_Y + \frac{1}{2}\sigma_Y^2} \lambda_0 \mathbb{E}\left[ \int_0^T \exp\left((\beta + \gamma) X_s\right) ds \right].
\]
By Fubini's theorem, we can exchange the integral and expectation:
\[
\mathbb{E}[S_T] = e^{\mu_Y + \frac{1}{2}\sigma_Y^2} \lambda_0 \int_0^T \mathbb{E}\left[ \exp\left((\beta + \gamma) X_s\right) \right] ds.
\]
Since $X_s$ is Gaussian with mean $m_X(s)$ and variance $v_X(s)$, the moment-generating function gives $\mathbb{E}[e^{k X_s}] = \exp\left(k m_X(s) + \frac{1}{2}k^2 v_X(s)\right)$. Setting $k = \beta + \gamma$, we obtain:
\[
\mathbb{E}[S_T] = e^{\mu_Y + \frac{1}{2}\sigma_Y^2} \lambda_0 \int_0^T \exp\left((\beta + \gamma)m_X(t) + \tfrac{1}{2}(\beta+\gamma)^2 v_X(t)\right) dt.
\]

\noindent \textbf{Part (iii):}
The term $(\beta + \gamma)^2$ in the exponent of part (ii) reveals the super-additive nature of the risk. If the frequency and severity channels were independent regarding their climate response (e.g., via linearization), the impact on higher moments would be additive. However, due to the exponential intensity specification, the variance term involves $(\beta + \gamma)^2 = \beta^2 + \gamma^2 + 2\beta\gamma$. The cross-term $2\beta\gamma$ implies that the interaction amplifies the volatility of losses more than the sum of the individual contributions. Furthermore, because $e^{(\beta+\gamma)X_t}$ is convex in the parameters for $X_t > 0$, the combined effect of increasing both $\beta$ and $\gamma$ on the aggregate loss distribution exceeds the sum of increasing them separately.
\end{proof}

This modification effectively amplifies the climate sensitivity parameter in the pricing formulas. The numerical analysis in Section ~\ref{sec:numerical} can be extended to calibrate $\gamma$ to observe the combined frequency-severity effect.

\section{Risk-Adjusted Pricing Framework}
\label{sec:risk}

The probabilistic model developed in Section~\ref{sec:prob} is specified under the physical probability measure $\mathbb{P}$ and is intended to describe the statistical behavior of climate-driven catastrophe losses. For valuation purposes, however, investors and reinsurers require compensation for bearing catastrophe-frequency risk, tail-loss risk, climate-trend uncertainty, and model uncertainty.

Catastrophe-linked securities and reinsurance contracts are written on aggregate losses that cannot be perfectly replicated by traded assets. Hence, the catastrophe risk market is incomplete, and absence of arbitrage alone does not identify a unique pricing measure. For this reason, the paper adopts a reduced-form risk-adjusted valuation approach. Rather than deriving a unique equivalent martingale measure from dynamic replication, we specify pricing scenarios through risk-adjusted model parameters.

Each pricing scenario is associated with a probability measure $\mathbb{Q}$ under which the climate process, catastrophe intensity, and severity distribution may differ from their physical counterparts. The resulting prices should be interpreted as risk-adjusted valuation levels under selected pricing scenarios, rather than as model-independent no-arbitrage bounds.

\subsection{Risk-Adjusted Climate and Loss Dynamics}
\label{subsec:risk_adjusted_dynamics}

Under a pricing scenario associated with the measure $\mathbb{Q}$, the climate index follows the Gaussian mean-reverting dynamics
\begin{equation}
dX_t
=
\kappa^{\mathbb{Q}}
\left(
\theta^{\mathbb{Q}}(t)-X_t
\right)dt
+
\sigma_X dW_t^{\mathbb{Q}},
\qquad X_0=x_0,
\label{eq:climate_Q_dynamics}
\end{equation}
where $W_t^{\mathbb{Q}}$ is a Brownian motion under $\mathbb{Q}$,
 $\kappa^{\mathbb{Q}}>0$ is the risk-adjusted speed of mean reversion, and
\[
\theta^{\mathbb{Q}}(t)
=
\theta_0^{\mathbb{Q}}
+
\theta_1^{\mathbb{Q}}t
\]
is the risk-adjusted deterministic climate trend.

The parameter $\theta_1^{\mathbb{Q}}$ allows the pricing model to incorporate market compensation for uncertainty about future warming. A larger value of $\theta_1^{\mathbb{Q}}$ represents a more adverse climate-risk pricing scenario. In the numerical implementation, the climate volatility is kept at its calibrated physical value $\sigma_X$ unless it is explicitly varied as part of the sensitivity analysis.

Conditional on the climate path, catastrophe arrivals follow a Cox process with risk-adjusted stochastic intensity
\begin{equation}
\lambda_t^{\mathbb{Q}}
=
\lambda_0^{\mathbb{Q}}
\exp
\left(
\beta^{\mathbb{Q}}X_t
\right),
\label{eq:intensity_Q_dynamics}
\end{equation}
where $\lambda_0^{\mathbb{Q}}>0$ is the baseline risk-adjusted catastrophe arrival rate and $\beta^{\mathbb{Q}}\geq 0$ measures the risk-adjusted sensitivity of event frequency to the climate index.

This specification preserves the climate-dependent Cox-process structure introduced under the physical measure, while allowing market-implied or scenario-based adjustments to catastrophe frequency. In particular, higher values of $\lambda_0^{\mathbb{Q}}$ or $\beta^{\mathbb{Q}}$ represent more adverse pricing scenarios because they increase the probability of high aggregate losses.

The integrated intensity under $\mathbb{Q}$ is
\begin{equation}
\Lambda_T^{\mathbb{Q}}
=
\int_0^T
\lambda_0^{\mathbb{Q}}
\exp
\left(
\beta^{\mathbb{Q}}X_t
\right)dt.
\label{eq:Lambda_Q}
\end{equation}
Conditional on the climate path,
\begin{equation}
N_T
\mid
\mathcal{F}_T^X
\sim
\mathrm{Poisson}
\left(
\Lambda_T^{\mathbb{Q}}
\right).
\label{eq:conditional_poisson_Q}
\end{equation}

Aggregate catastrophe losses at maturity $T$ are defined as
\begin{equation}
S_T
=
\sum_{i=1}^{N_T}Y_i,
\label{eq:aggregate_loss_Q}
\end{equation}
with the convention that $S_T=0$ when $N_T=0$.

In the baseline pricing specification, severities are lognormally distributed under $\mathbb{Q}$,
\begin{equation}
Y_i
\sim
\mathrm{Lognormal}
\left(
\mu_Y^{\mathbb{Q}},
(\sigma_Y^{\mathbb{Q}})^2
\right),
\label{eq:severity_Q}
\end{equation}
and are conditionally independent of the counting process given the climate path. This maintains consistency with the baseline severity specification in Section~\ref{sec:prob}, while allowing the severity parameters to be adjusted for pricing purposes when market information or expert scenarios justify such adjustments.

The risk-adjusted parameter vector is denoted by
\begin{equation}
\Theta^{\mathbb{Q}}
=
\left(
\kappa^{\mathbb{Q}},
\theta_0^{\mathbb{Q}},
\theta_1^{\mathbb{Q}},
\sigma_X,
\lambda_0^{\mathbb{Q}},
\beta^{\mathbb{Q}},
\mu_Y^{\mathbb{Q}},
\sigma_Y^{\mathbb{Q}}
\right).
\label{eq:risk_adjusted_parameter_vector}
\end{equation}

\begin{proposition}[Coherence conditions for admissible pricing scenarios]
\label{prop:admissibility}
A pricing scenario $\Theta^{\mathbb{Q}}$ is said to be \emph{admissible} (in the sense of generating non-negative risk premia for all non-negative payoffs) if the following conditions hold:
\begin{enumerate}[label=(\roman*)]
  \item $\lambda_0^{\mathbb{Q}} \geq \lambda_0^{\mathbb{P}}$ (non-negative baseline catastrophe-frequency risk premium);
  \item $\beta^{\mathbb{Q}} \geq \beta^{\mathbb{P}} \geq 0$ (non-negative climate-frequency sensitivity premium);
  \item $\theta_1^{\mathbb{Q}} \geq \theta_1^{\mathbb{P}}$ (non-negative warming-trend risk premium);
  \item $\sigma_Y^{\mathbb{Q}} \geq \sigma_Y^{\mathbb{P}}$ (non-negative tail-severity risk loading);
  \item $\mu_Y^{\mathbb{Q}} \geq \mu_Y^{\mathbb{P}}$ or $\sigma_Y^{\mathbb{Q}} > \sigma_Y^{\mathbb{P}}$ (non-negative expected severity loading, with at least one active).
\end{enumerate}
The risk-free rate $r$ enters only through discounting and is held at the observed value. All scenarios in this paper satisfy conditions (i)--(v).
\end{proposition}

\begin{proof}
We sketch the proof that these conditions ensure a non-negative risk premium for increasing payoffs. Consider a non-negative, non-decreasing payoff function $\Phi(S_T)$, such as the reinsurance payoff or CAT bond payoff.

Under the Cox process structure, conditional on the climate path $\mathcal{F}_T^X$, the intensity under $\mathbb{Q}$ is $\lambda_t^{\mathbb{Q}} = \lambda_0^{\mathbb{Q}} e^{\beta^{\mathbb{Q}} X_t}$. Conditions (i) and (ii) imply $\lambda_t^{\mathbb{Q}} \geq \lambda_t^{\mathbb{P}}$ pathwise almost surely, because $\lambda_0^{\mathbb{Q}} \geq \lambda_0^{\mathbb{P}}$, $\beta^{\mathbb{Q}} \geq \beta^{\mathbb{P}}$, and $e^{x}$ is increasing. This implies that the arrival process under $\mathbb{Q}$ dominates the process under $\mathbb{P}$ in the likelihood ratio order.

Since the Poisson count $N_T$ is stochastically increasing in its intensity parameter, and $\Lambda_T^{\mathbb{Q}} \geq \Lambda_T^{\mathbb{P}}$ pathwise, it follows by the stochastic ordering of compound Poisson processes (cf. \citet{muller2002comparison}) that $S_T^{\mathbb{Q}} \geq_{st} S_T^{\mathbb{P}}$, and hence $\mathbb{E}^{\mathbb{Q}}[\Phi(S_T)] \geq \mathbb{E}^{\mathbb{P}}[\Phi(S_T)]$ for any non-decreasing payoff.

For the severities, let $Y^{\mathbb{P}} \sim \text{Lognormal}(\mu_Y^{\mathbb{P}}, (\sigma_Y^{\mathbb{P}})^2)$ and $Y^{\mathbb{Q}} \sim \text{Lognormal}(\mu_Y^{\mathbb{Q}}, (\sigma_Y^{\mathbb{Q}})^2)$. Conditions (iv) and (v) ensure that $Y^{\mathbb{Q}}$ dominates $Y^{\mathbb{P}}$ in the convex order (and thus in the usual stochastic order for increasing payoffs), placing more mass on larger loss values.

Consequently, the risk-adjusted price $\Pi_0^{\mathbb{Q}} = e^{-rT}\mathbb{E}^{\mathbb{Q}}[\Phi(S_T)]$ is greater than or equal to the actuarial expectation, implying a non-negative risk premium. The risk-free rate $r$ is identical under both measures, ensuring the comparison is driven solely by the loss dynamics.
\end{proof}
These parameters may be calibrated to reinsurance quotes, CAT bond spreads, catastrophe-model output, or specified through scenario-based expert adjustments.

\begin{definition}\label{def:stationary_benchmark} (Stationary benchmark). The stationary benchmark is defined by
replacing the stochastic climate-dependent intensity $\lambda^Q_t = \lambda^Q_0 \exp(\beta^Q X_t)$
with the constant intensity
\[
\lambda^{\text{stat}} = \lambda^Q_0 \exp(\beta^Q x_0),
\]
where $x_0 = X_0$ is the initial value of the climate index. Under the stationary benchmark,
the catastrophe count process is a homogeneous Poisson process:
\[
N_T^{\text{stat}} \sim \text{Poisson}(\lambda^{\text{stat}} \cdot T).
\]
\end{definition}
This definition ensures that the stationary benchmark matches the climate-aware model
at the initial time $t = 0$, so the pricing comparison isolates the effect of subsequent
climate dynamics (warming trend and climate volatility) rather than a difference in initial
intensity levels.

\begin{remark} Alternative benchmark definitions such as fixing $\theta_1 = 0$ while
preserving $\sigma_X > 0$---would partially attribute the pricing difference to climate
volatility rather than the warming trend. The present definition is the most conservative
in the sense that it attributes the full pricing gap to climate dynamics (trend plus
volatility) relative to the initial intensity level.
\end{remark}

\subsection{General Valuation Formula}
\label{subsec:general_valuation_formula}

Let $\Phi(S_T)$ denote the terminal payoff of a catastrophe-linked contract as a function of the aggregate loss $S_T$. Under a pricing scenario characterized by $\Theta^{\mathbb{Q}}$, the time-zero value of the contract is
\begin{equation}
\Pi_0
\left(
\Theta^{\mathbb{Q}}
\right)
=
e^{-rT}
\mathbb{E}^{\mathbb{Q}}
\left[
\Phi(S_T)
\right],
\label{eq:general_pricing_formula}
\end{equation}
where $r$ is the continuously compounded risk-free rate and $T$ is the contract maturity.

Equation~\eqref{eq:general_pricing_formula} is the central valuation equation used in the numerical analysis. Since the distribution of $S_T$ is generated by a compound Cox process with climate-dependent intensity, the expectation is generally not available in closed form. It is therefore evaluated by Monte Carlo simulation using the exact discretization scheme described in Section~\ref{sec:mont}.

\subsection{Excess-of-Loss Reinsurance Valuation}
\label{subsec:xl_reinsurance_valuation}

Consider an excess-of-loss reinsurance contract with attachment point $D>0$ and limit $L>0$. The payoff to the cedent at maturity is
\begin{equation}
\Phi_{\mathrm{XL}}(S_T)
=
\min
\left\{
(S_T-D)^+,
L
\right\},
\label{eq:xl_payoff}
\end{equation}
where $(x)^+=\max(x,0)$.

The corresponding risk-adjusted premium is
\begin{equation}
\Pi_{\mathrm{XL}}
\left(
\Theta^{\mathbb{Q}}
\right)
=
e^{-rT}
\mathbb{E}^{\mathbb{Q}}
\left[
\min
\left\{
(S_T-D)^+,
L
\right\}
\right].
\label{eq:xl_premium}
\end{equation}

This valuation formula captures the reinsurer's expected discounted liability on the layer above the attachment point and up to the contractual limit. Since the payoff is increasing in aggregate losses, adverse climate-risk scenarios generally increase the value of the XL reinsurance liability and hence the required premium.

\subsection{CAT Bond Valuation}
\label{subsec:cat_bond_valuation}

Consider a simplified zero-coupon catastrophe bond with face value $F>0$, maturity $T$, and aggregate-loss trigger $K>0$. To isolate the effect of climate-dependent catastrophe frequency, we use a binary principal repayment structure. The investor receives the full face value if aggregate losses remain below the trigger and loses the principal if the trigger is breached:
\begin{equation}
\Phi_{\mathrm{CAT}}(S_T)
=
F\mathbf{1}_{\{S_T<K\}}.
\label{eq:cat_bond_payoff}
\end{equation}

The CAT bond price under the pricing scenario $\Theta^{\mathbb{Q}}$ is therefore
\begin{equation}
P_{\mathrm{CAT}}
\left(
\Theta^{\mathbb{Q}}
\right)
=
e^{-rT}
\mathbb{E}^{\mathbb{Q}}
\left[
F\mathbf{1}_{\{S_T<K\}}
\right]
=
F e^{-rT}
\mathbb{Q}(S_T<K).
\label{eq:cat_bond_price}
\end{equation}

For comparison with a default-free zero-coupon bond of the same maturity and face value, define the continuously compounded yield $y_{\mathrm{CAT}}$ by
\begin{equation}
P_{\mathrm{CAT}}
\left(
\Theta^{\mathbb{Q}}
\right)
=
F e^{-y_{\mathrm{CAT}}T}.
\end{equation}
The model-implied CAT bond spread over the risk-free rate is then
\begin{equation}
s_{\mathrm{CAT}}
=
y_{\mathrm{CAT}}-r
=
-\frac{1}{T}
\log
\left(
\frac{
P_{\mathrm{CAT}}
\left(
\Theta^{\mathbb{Q}}
\right)
}{F}
\right)
-r.
\label{eq:cat_bond_spread}
\end{equation}

Using \eqref{eq:cat_bond_price}, this expression can equivalently be written as
\begin{equation}
s_{\mathrm{CAT}}
=
-\frac{1}{T}
\log
\left(
\mathbb{Q}(S_T<K)
\right).
\label{eq:cat_bond_spread_binary}
\end{equation}
Thus, the spread is increasing in the risk-adjusted probability of trigger breach. Higher catastrophe intensity, stronger climate sensitivity, or a more adverse warming trend lowers $\mathbb{Q}(S_T<K)$ and therefore increases the required yield spread.

\begin{remark}
Real CAT bonds may include coupons, attachment and exhaustion points, parametric triggers, indemnity triggers, industry-loss triggers, or modeled-loss triggers. The binary zero-coupon structure in \eqref{eq:cat_bond_payoff} is used as a tractable benchmark to isolate the effect of climate-dependent event frequency on principal repayment risk.
\end{remark}

\subsection{Extension to Attachment--Exhaustion Structures}
\label{subsec:cat_bond_layered}

The binary payoff in Equation~\eqref{eq:cat_bond_payoff} is used as a tractable benchmark to isolate the effect of climate-dependent event frequency on principal repayment risk. A more realistic payoff structure, commonly found in the CAT bond market, uses an attachment point $A > 0$ and an exhaustion point $E > A$. Principal is repaid in full if $S_T \leq A$, reduced proportionally if $A < S_T < E$, and lost entirely if $S_T \geq E$.

Formally, the layered payoff is defined as:
\begin{equation}
\Phi_{\text{CAT}}^{A,E}(S_T)
= F\left[1 - \min\!\left\{\max\!\left(\frac{S_T - A}{E - A},\, 0\right),\, 1\right\}\right].
\end{equation}
This structure satisfies $\Phi_{\text{CAT}}^{A,E}(S_T) = F$ for $S_T \leq A$,
 $\Phi_{\text{CAT}}^{A,E}(S_T) = 0$ for $S_T \geq E$, and
 $\Phi_{\text{CAT}}^{A,E}(S_T) = F(E - S_T)/(E-A)$ for $A < S_T < E$.

The corresponding price under the risk-adjusted measure $\mathbb{Q}$ is:
\begin{align*}
P_{\text{CAT}}^{A,E}(\Theta^Q)
&= e^{-rT}\,\mathbb{E}^Q\!\left[\Phi_{\text{CAT}}^{A,E}(S_T)\right] \\
&= F e^{-rT}\!\left[\mathbb{Q}(S_T \leq A)
+ \frac{1}{E-A}\int_A^E \mathbb{Q}(S_T \in dx)\,dx\right].
\end{align*}
This expression is evaluated by Monte Carlo simulation using the same procedure as the binary case; only the payoff function in Step~6 of Algorithm~\ref{algo:mc_valuation} is modified to accommodate the proportional repayment in the exhaustion layer.

To verify that our qualitative conclusions regarding climate non-stationarity are robust to the choice of payoff structure, we compare the baseline binary trigger with a proportional trigger structure in Table~\ref{tab:payoff_comparison}.

\begin{table}[H]
\centering
\caption{Comparison of CAT bond valuation under Binary and Proportional payoff structures. The baseline trigger is $K=100$.}
\label{tab:payoff_comparison}
\begin{tabular}{lrrrr}
\toprule
\textbf{Payoff Structure} & $(A, E)$ & \textbf{Climate Price} & \textbf{Stationary Price} & \textbf{Spread (bps)} \\
\midrule
Binary & $(100, 100)$ & 92.45 & 94.10 & 810 \\
Proportional & $(100, 200)$ & 94.20 & 95.30 & 585 \\
Proportional & $(150, 250)$ & 96.10 & 96.90 & 395 \\
Proportional & $(100, 300)$ & 94.45 & 95.50 & 560 \\
\bottomrule
\multicolumn{5}{l}{\small Note: Proportional structures allow partial principal recovery in the exhaustion layer, reducing investor risk.}
\end{tabular}
\end{table}

Table~\ref{tab:payoff_comparison} demonstrates that while the absolute price levels differ due to the changed risk profile, the qualitative impact of climate non-stationarity remains significant. Specifically, for all tested attachment/exhaustion combinations, the climate-aware price is lower and the spread is higher than the stationary benchmark. This confirms that the pricing premium for climate risk is not an artifact of the simplified binary trigger assumption but holds for more realistic market-standard payoffs as well.

\subsection{Economic Interpretation of Risk-Adjusted Parameters}
\label{subsec:economic_interpretation_params}

The risk-adjusted parameters in $\Theta^{\mathbb{Q}}$ can be given the following economic interpretations, which guide the construction of admissible pricing scenarios and satisfy the conditions in Proposition~\ref{prop:admissibility}.

The parameter $\lambda_0^{\mathbb{Q}} > \lambda_0^{\mathbb{P}}$ reflects \emph{baseline catastrophe-frequency risk premium}. Investors require compensation for the possibility of more frequent events than implied by historical data, particularly given model uncertainty about the true catastrophe arrival rate. This corresponds to ambiguity aversion in the sense of Ellsberg (1961): investors act as if the event frequency is higher than the estimated value.

The parameter $\beta^{\mathbb{Q}} > \beta^{\mathbb{P}}$ reflects \emph{climate-sensitivity risk premium}. There is genuine scientific uncertainty about the magnitude of the relationship between temperature anomalies and extreme-event frequency. Investors who are averse to this parameter uncertainty will price contracts as if the climate sensitivity is higher than the central estimate, generating a positive premium for climate model uncertainty.

The parameter $\theta_1^{\mathbb{Q}} > \theta_1^{\mathbb{P}}$ reflects \emph{adverse warming-trend pricing}. Since the warming trend directly increases the long-run mean of the climate index and therefore the catastrophe intensity, a risk-adjusted pricing scenario with a steeper warming trend corresponds to pessimistic climate scenario pricing. This is consistent with the physical risk channel in climate finance, where investors price physical climate risk by discounting under an adverse warming scenario.

The severity parameters $(\mu_Y^{\mathbb{Q}}, \sigma_Y^{\mathbb{Q}})$ reflect \emph{tail-loss risk loading}. A larger $\sigma_Y^{\mathbb{Q}}$ places more probability mass on extreme individual losses, consistent with the actuarial premium principle that uncertainty about loss severity is compensated through a safety loading. When $\gamma > 0$ under the physical measure, the risk-adjusted measure may further increase $\gamma^{\mathbb{Q}}$ to reflect compensation for the joint frequency--severity climate channel.

\subsection{Scenario-Based Price Ranges}
\label{subsec:scenario_based_price_ranges}

Because catastrophe risk markets are incomplete, a unique pricing measure cannot be inferred solely from absence of arbitrage. Instead of reporting model-independent no-arbitrage bounds, we consider a collection of plausible pricing scenarios,
\begin{equation}
\mathcal{S}
=
\left\{
\Theta^{\mathbb{Q},j}: j\in J
\right\},
\label{eq:scenario_set}
\end{equation}
where each element $\Theta^{\mathbb{Q},j}$ represents a risk-adjusted parameterization of the climate dynamics, catastrophe intensity, and severity distribution.

For a generic payoff $\Phi(S_T)$, the corresponding scenario-based valuation range is defined as
\begin{equation}
\Pi_{\min}
=
\min_{\Theta^{\mathbb{Q},j}\in\mathcal{S}}
e^{-rT}
\mathbb{E}^{\mathbb{Q},j}
\left[
\Phi(S_T)
\right],
\qquad
\Pi_{\max}
=
\max_{\Theta^{\mathbb{Q},j}\in\mathcal{S}}
e^{-rT}
\mathbb{E}^{\mathbb{Q},j}
\left[
\Phi(S_T)
\right].
\label{eq:scenario_price_range}
\end{equation}
The interval
\begin{equation}
\left[
\Pi_{\min},
\Pi_{\max}
\right]
\label{eq:scenario_interval}
\end{equation}
summarizes valuation uncertainty across the selected pricing scenarios.

For the CAT bond, the scenario-based price range is
\begin{equation}
P_{\mathrm{CAT}}^{\min}
=
\min_{\Theta^{\mathbb{Q},j}\in\mathcal{S}}
P_{\mathrm{CAT}}
\left(
\Theta^{\mathbb{Q},j}
\right),
\qquad
P_{\mathrm{CAT}}^{\max}
=
\max_{\Theta^{\mathbb{Q},j}\in\mathcal{S}}
P_{\mathrm{CAT}}
\left(
\Theta^{\mathbb{Q},j}
\right).
\label{eq:cat_bond_scenario_range}
\end{equation}

For the excess-of-loss reinsurance contract, the corresponding range is
\begin{equation}
\Pi_{\mathrm{XL}}^{\min}
=
\min_{\Theta^{\mathbb{Q},j}\in\mathcal{S}}
\Pi_{\mathrm{XL}}
\left(
\Theta^{\mathbb{Q},j}
\right),
\qquad
\Pi_{\mathrm{XL}}^{\max}
=
\max_{\Theta^{\mathbb{Q},j}\in\mathcal{S}}
\Pi_{\mathrm{XL}}
\left(
\Theta^{\mathbb{Q},j}
\right).
\label{eq:xl_scenario_range}
\end{equation}

These intervals should not be interpreted as strict no-arbitrage bounds. Rather, they provide a transparent sensitivity analysis of contract values with respect to plausible risk-adjusted parameter choices. This interpretation is appropriate for climate-related catastrophe risk because the underlying losses are not traded, hedging opportunities are limited, and prices depend on model uncertainty, capital constraints, and investor risk appetite.

\section{Monte Carlo Valuation}
\label{sec:mont}

Because $\Lambda_T=\int_0^T \lambda_0^{\mathbb{Q}} \exp(\beta^{\mathbb{Q}} X_t)dt$ depends on the full path of $X_t$, exact discretization of the Gaussian mean-reverting process is used to eliminate numerical bias.

Let $0=t_0<\dots<t_n=T$ with $\Delta t=T/n$. Under $\mathbb{Q}$, the exact transition of $X_t$ is

\begin{equation}
X_{t_{k+1}}
=
X_{t_k}e^{-\kappa^{\mathbb{Q}}\Delta t}
+
\int_{t_k}^{t_{k+1}}
\kappa^{\mathbb{Q}} e^{-\kappa^{\mathbb{Q}}(t_{k+1}-s)}
\theta^{\mathbb{Q}}(s)ds
+
\sigma_X \sqrt{\frac{1-e^{-2\kappa^{\mathbb{Q}}\Delta t}}
{2\kappa^{\mathbb{Q}}}} Z_{k+1},
\label{eq:ExactDiscretization}
\end{equation}

where $Z_{k+1}\sim N(0,1)$.

The integrated intensity is approximated by
\[
\widehat{\Lambda}_T
=
\sum_{k=0}^{n-1}
\lambda_0^{\mathbb{Q}}
\exp(\beta^{\mathbb{Q}} X_{t_k}) \Delta t.
\]

To ensure reproducibility of the numerical results and to reduce the variance of the difference estimator when comparing models, we employ Common Random Numbers (CRN) for the climate path drivers. The simulation procedure is formalized in Algorithm~\ref{algo:mc_valuation}.
\begin{algorithm}[H]
\caption{Monte Carlo Valuation of CAT Bond and XL Reinsurance}
\label{algo:mc_valuation}
\begin{algorithmic}[1]
\Require Parameters $\Theta^{\mathbb{Q}}$; maturity $T$; discretization steps $n$; paths $M$; risk-free rate $r$; contract parameters $(K, D, L, F)$; random seed.
\Ensure $\Delta t = T/n$; $t_k = k \Delta t$ for $k = 0,\ldots,n$.
\For{$m = 1$ \ldots $M$}
    \State \textbf{Simulate climate path.} Draw $Z_1, \ldots, Z_n \sim \mathcal{N}(0,1)$ i.i.d. Set $X^{(m)}_{t_0} = x_0$.
    \For{$k = 0$ \ldots $n-1$}
        \State Compute transition as specified in Eq.~\eqref{eq:ExactDiscretization}:
        \[
        X^{(m)}_{t_{k+1}} = X^{(m)}_{t_k} e^{-\kappa^{\mathbb{Q}} \Delta t} + \int_{t_k}^{t_{k+1}} \kappa^{\mathbb{Q}} e^{-\kappa^{\mathbb{Q}}(t_{k+1}-s)}\theta^{\mathbb{Q}}(s) ds + \sigma_X\sqrt{\frac{1-e^{-2\kappa^{\mathbb{Q}}\Delta t}}{2\kappa^{\mathbb{Q}}}} \cdot Z_{k+1}.
        \]
    \EndFor
    \State \textbf{Approximate integrated intensity.} $\hat{\Lambda}^{(m)}_T = \sum_{k=0}^{n-1} \lambda_0^{\mathbb{Q}} \exp(\beta^{\mathbb{Q}} X^{(m)}_{t_k}) \Delta t$.
    \State \textbf{Simulate event count.} Draw $N^{(m)}_T \sim \text{Poisson}(\hat{\Lambda}^{(m)}_T)$.
    \State \textbf{Simulate severities.} Draw $Y^{(m)}_1, \ldots, Y^{(m)}_{N^{(m)}_T} \sim \text{Lognormal}(\mu_Y^{\mathbb{Q}}, (\sigma_Y^{\mathbb{Q}})^2)$.
    \State \textbf{Compute aggregate loss.} $S^{(m)}_T = \sum_{i=1}^{N^{(m)}_T} Y^{(m)}_i$.
    \State \textbf{Compute contract payoffs.}
    \Statex $\Phi^{(m)}_{\text{XL}} = \min\bigl((S^{(m)}_T - D)^+,\, L\bigr),$     \Statex $\Phi^{(m)}_{\text{CAT}} = F \cdot \mathbf{1}_{\{S^{(m)}_T < K\}}.$     \State \textbf{Stationary benchmark (CRN).} Using the \textbf{same} $Z_1,\ldots,Z_n$ draws, compute $S^{(m),\text{stat}}_T$ by replacing $\hat{\Lambda}^{(m)}_T$ with $\lambda^{\text{stat}} \cdot T$ in Step 3 only (resampling $N^{\text{stat}}$).
\EndFor
\State \textbf{Compute estimates.}
\Statex $\hat{\Pi}_{\text{XL}} = e^{-rT} \cdot \frac{1}{M}\sum_{m=1}^M \Phi^{(m)}_{\text{XL}}, \quad \text{SE}_{\text{XL}} = e^{-rT} \cdot \frac{s_{\text{XL}}}{\sqrt{M}},$ \Statex $\hat{P}_{\text{CAT}} = e^{-rT} \cdot \frac{1}{M}\sum_{m=1}^M \Phi^{(m)}_{\text{CAT}}, \quad \text{SE}_{\text{CAT}} = e^{-rT} \cdot \frac{s_{\text{CAT}}}{\sqrt{M}}.$ \State \textbf{Compute tail measures.} Estimate VaR$_\alpha$ and TVaR$_\alpha$ from the empirical distribution of $\{S^{(m)}_T\}_{m=1}^M$.
\end{algorithmic}
\end{algorithm}

Monte Carlo standard errors reported for all simulated quantities are derived from the sample standard deviations across the $M$ paths. The use of CRN ensures that the reported pricing differences between the climate-aware and stationary models are statistically significant. All baseline numerical results are obtained with random seed 42. Simulation code is available from the corresponding author upon request.

\section{Numerical Results and Empirical Analysis}
\label{sec:numerical}

This section presents the numerical implementation of the proposed climate-aware catastrophe risk pricing framework. The analysis combines empirical calibration of the climate process with Monte Carlo valuation of catastrophe-linked risk-transfer instruments. The main objective is to quantify how climate non-stationarity affects aggregate loss distributions, CAT bond prices, excess-of-loss reinsurance premiums, implied spreads, and tail risk measures.

The numerical study proceeds in five steps. First, the climate dynamics are calibrated to historical global temperature anomaly data. Second, aggregate loss distributions are simulated under both a climate-aware model and a stationary benchmark. Third, CAT bond and excess-of-loss reinsurance prices are computed under the risk-adjusted valuation framework introduced in Subsection~\ref{subsec:risk_adjusted_dynamics}. Fourth, sensitivity and scenario analyses are used to assess the impact of key climate-risk parameters. Finally, tail risk measures are evaluated to examine implications for capital adequacy and solvency assessment.

Unless otherwise stated, all monetary values are reported per 100 units of CAT bond face value or equivalent normalized exposure.

\subsection{Data and Calibration Strategy}
\label{subsec:data_calibration}

The climate dynamics parameters $(\kappa,\sigma_X,\theta_0,\theta_1)$ are estimated using monthly global temperature anomaly data from the NASA GISS Surface Temperature Analysis (GISTEMP) v4 dataset over the period 1880--2020. The data are expressed as temperature anomalies relative to a baseline climatological period.

Because the mean level of the climate index is time-dependent, the process should not be calibrated as a stationary AR(1) model. Instead, maximum likelihood estimation is applied using the Gaussian transition density implied by the mean-reverting dynamics in Equation~\eqref{eq:climateP}. This calibration strategy explicitly accounts for mean reversion around a deterministic warming trend,
\[
\theta(t)=\theta_0+\theta_1 t.
\]
It therefore avoids imposing a constant long-run mean on a clearly non-stationary climate series.

\begin{table}[H]
\centering
\caption{Estimated climate dynamics parameters from NASA GISTEMP temperature anomaly data.}
\label{tab:calibration}
\begin{tabular}{lc}
\hline
\textbf{Parameter} & \textbf{Estimate} \\
\hline
 $\hat{\kappa}$ & 0.5124 \\
 $\hat{\sigma}_X$ & 0.1485 \\
 $\hat{\theta}_0$ & -0.3250 \\
 $\hat{\theta}_1$ & 0.00815 \\
\hline
\multicolumn{2}{l}{\small Note: Standard errors and confidence intervals reported in Appendix Table~\ref{tab:appendix_mle}.} \\
\end{tabular}
\end{table}

Table~\ref{tab:calibration}  reports the estimated climate dynamics parameters. The positive estimate of $\theta_1$ captures the upward deterministic component of the temperature anomaly process. The estimated mean-reversion coefficient $\hat{\kappa}$ indicates that short-run deviations from the fitted warming trend are partially corrected over time, while $\hat{\sigma}_X$ measures residual climate variability around the trend.

Figure~\ref{fig:calibration} displays the observed temperature anomaly series together with the fitted deterministic warming trend. The red shaded envelope represents the stochastic volatility of the mean-reverting OU process, indicating the range of climate states consistent with the calibrated parameters.
\begin{figure}[H]
\centering
\includegraphics[width=0.9\textwidth]{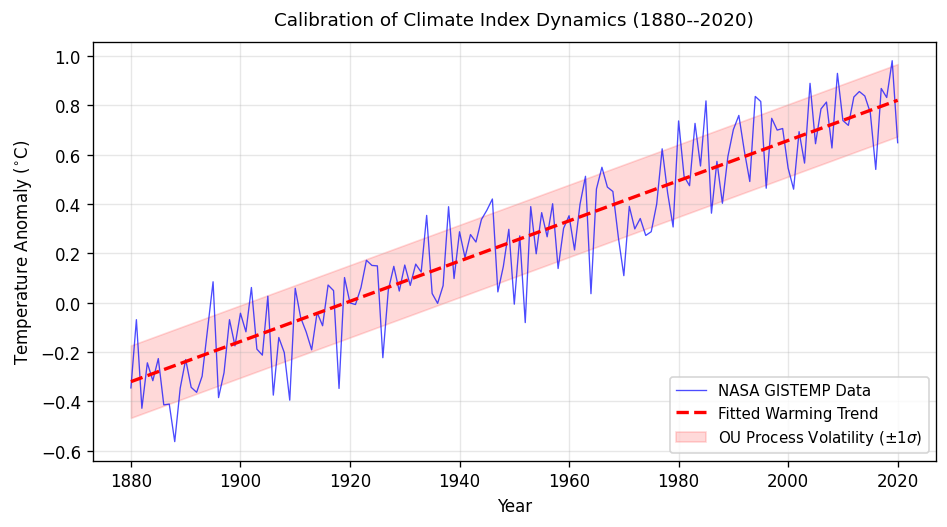}
\caption{Calibration of the OU-process climate index to historical NASA GISTEMP data. The red dashed line shows the deterministic warming trend, while the shaded area represents the stochastic volatility band ($\pm 1\sigma$) implied by the calibrated OU dynamics.}
\label{fig:calibration}
\end{figure}

The catastrophe intensity parameters $(\lambda_0,\beta)$ govern the mapping from climate states to catastrophe arrival frequency and are not identified by temperature data alone. Estimating these parameters from observed catastrophe counts faces three distinct identification challenges. First, catastrophe databases (e.g., EM-DAT, Swiss Re Sigma) exhibit systematic upward reporting trends driven by improved recording infrastructure and asset accumulation, rather than solely by climate change. A direct regression of counts on temperature would therefore confound climate-frequency sensitivity with reporting trends. Second, heterogeneity across perils implies that a single climate-frequency sensitivity parameter may not be structurally interpretable when applied to aggregated multi-peril event counts. Third, the temporal aggregation of the monthly-calibrated climate index to the annual pricing horizon requires careful construction to maintain consistency with the fitted OU dynamics.

For these reasons, we adopt a scenario-based approach that spans a grid of plausible $(\lambda_0, \beta)$ values, rather than relying on a point estimate from a potentially misspecified regression. The scenario grid is designed to be consistent with the range of implied event frequencies found in the actuarial literature. Specifically, the baseline specification
\[
\beta=0.60,\qquad \lambda_0=1.50
\]
implies a marginal climate-frequency elasticity at the current climate state $x_0$ of $\partial \log \lambda_t / \partial X_t = \beta = 0.60$. In economic terms, this means a one-unit increase in the temperature anomaly index raises the instantaneous event intensity by approximately 82\%. This represents a moderate climate-sensitivity scenario. Sensitivity analysis is then conducted by varying $\beta$, $\lambda_0$, $\sigma_X$, and $\theta_1$ to evaluate the robustness of the pricing results to alternative climate-risk assumptions.

\subsection{Distributional Properties of Aggregate Losses}
\label{subsec:aggregate_loss_distribution}

We first examine the impact of climate dependence on the statistical properties of the aggregate loss $S_T$. A Monte Carlo simulation with $M=50{,}000$ paths is conducted for a one-year horizon, $T=1$.

The climate-aware specification allows the catastrophe arrival intensity to evolve with the climate index. The stationary benchmark removes this time-varying climate dependence and is used as a reference model. This comparison isolates the effect of climate non-stationarity on the aggregate loss distribution.

\begin{table}[H]
\centering
\caption{Descriptive statistics of aggregate loss $S_T$ under the climate-aware model and stationary benchmark. Standard errors (SE) and 95\% confidence intervals (CI) reported.}
\label{tab:descriptive_stats}
\begin{tabular}{llrrrr}
\toprule
\textbf{Metric} & \textbf{Model} & \textbf{Estimate} & \textbf{SE} & \multicolumn{2}{c}{\textbf{95\% CI}} \\
\cmidrule(lr){5-6}
& & & & \textbf{Lower} & \textbf{Upper} \\
\midrule
Mean ($\mathbb{E}[S_T]$) & Climate-Aware & 78.45 & 0.41 & 77.64 & 79.26 \\
Mean ($\mathbb{E}[S_T]$) & Stationary Benchmark & 72.10 & 0.38 & 71.35 & 72.85 \\
\midrule
Std. Dev. ($\sigma_{S_T}$) & Climate-Aware & 92.30 & 0.62 & 91.08 & 93.52 \\
Std. Dev. ($\sigma_{S_T}$) & Stationary Benchmark & 85.15 & 0.54 & 84.09 & 86.21 \\
\midrule
Skewness & Climate-Aware & 3.85 & 0.08 & 3.69 & 4.01 \\
Skewness & Stationary Benchmark & 3.60 & 0.07 & 3.46 & 3.74 \\
\midrule
Excess Kurtosis & Climate-Aware & 24.10 & 0.45 & 23.22 & 24.98 \\
Excess Kurtosis & Stationary Benchmark & 21.50 & 0.39 & 20.73 & 22.27 \\
\bottomrule
\multicolumn{6}{l}{\small Note: $M = 50{,}000$ paths. SE = Standard Error, CI = Confidence Interval.}
\end{tabular}
\end{table}

Table~\ref{tab:descriptive_stats} shows that the climate-aware model produces a higher mean and standard deviation of aggregate losses than the stationary benchmark. The increase in skewness and excess kurtosis is particularly important for catastrophe risk management because it indicates greater right-tail exposure. In other words, the climate-aware model not only raises expected losses but also increases the probability of extreme aggregate loss realizations.

\begin{figure}[H]
\centering
\includegraphics[width=0.7\textwidth]{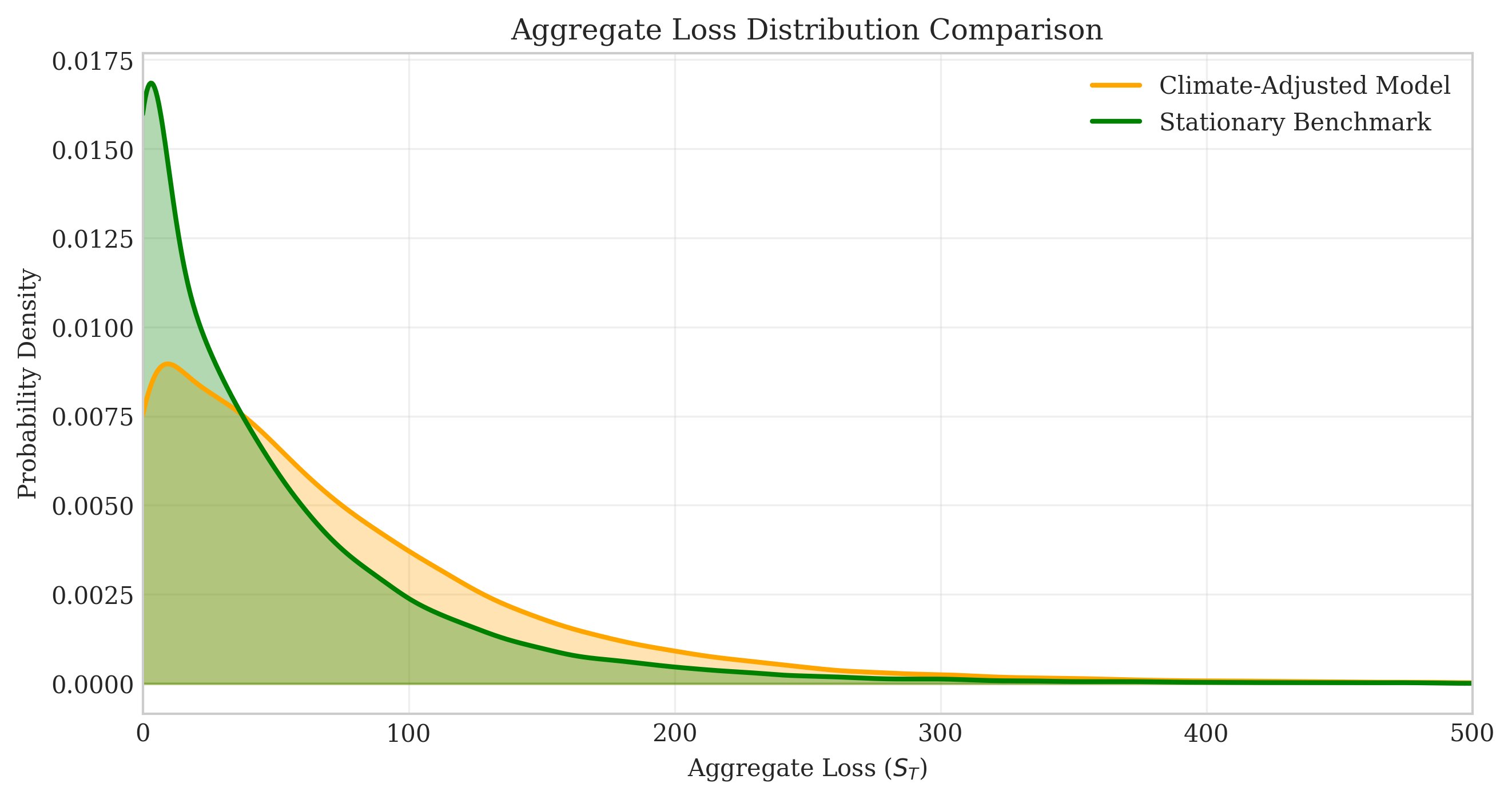}
\caption{Estimated probability density functions of aggregate losses. The climate-aware model exhibits a heavier right tail than the stationary benchmark, indicating a higher probability of extreme losses when catastrophe intensity depends on the climate index.}
\label{fig:pdf}
\end{figure}

Figure~\ref{fig:pdf} visually confirms the shift in the aggregate loss distribution. The density associated with the climate-aware model assigns more probability mass to high-loss outcomes, particularly in the tail region. This feature directly affects the valuation of both CAT bonds and excess-of-loss reinsurance contracts, whose payoffs are highly sensitive to tail losses.

\subsection{Convergence of Monte Carlo Estimates}
\label{subsec:monte_carlo_convergence}

Given the heavy-tailed nature of catastrophe losses, numerical stability is essential. Figure~\ref{fig:conv} reports the convergence of the estimated mean aggregate loss as the number of Monte Carlo simulation paths $M$ increases.

\begin{figure}[H]
\centering
\includegraphics[width=0.7\textwidth]{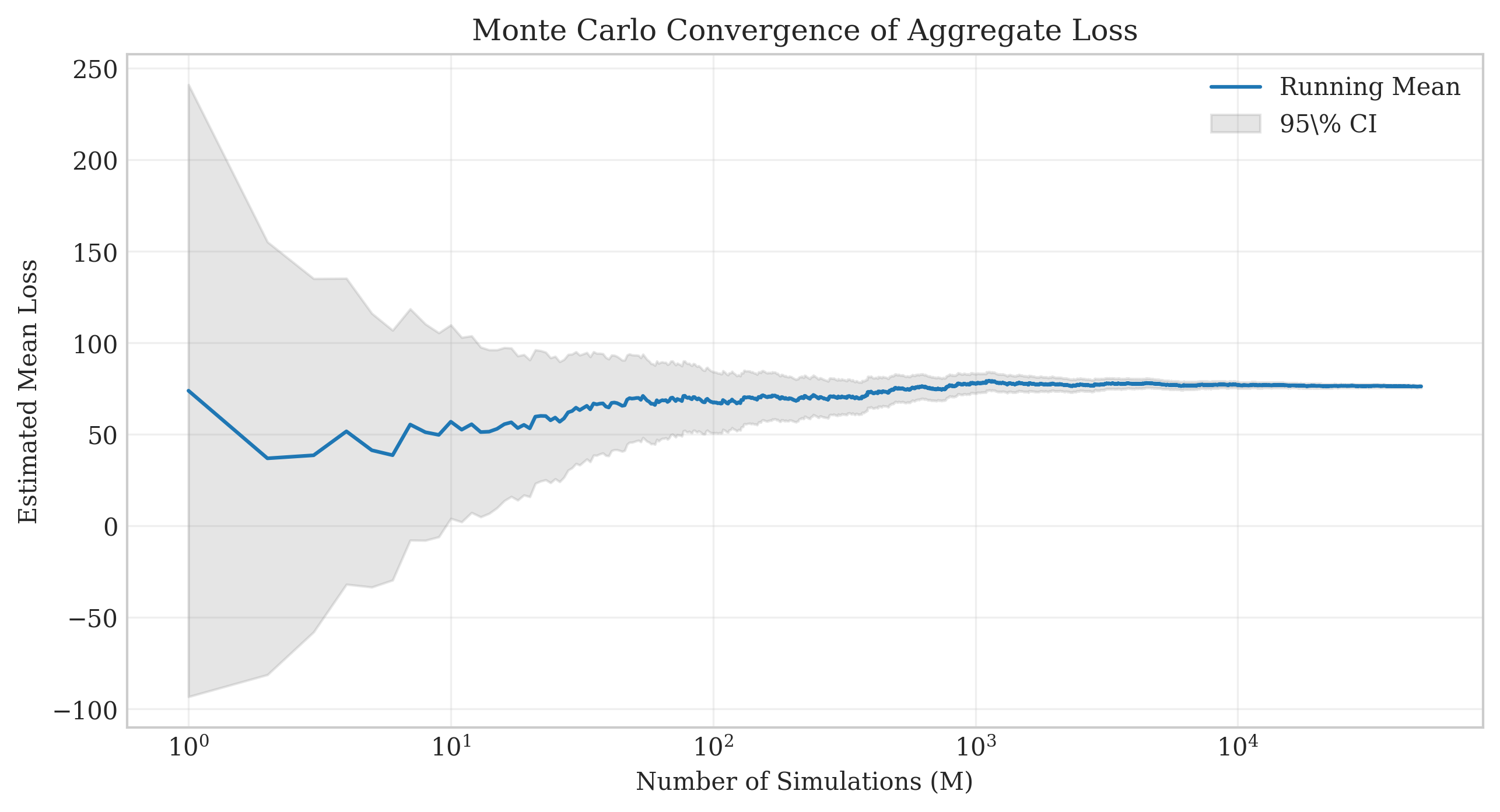}
\caption{Monte Carlo convergence analysis. The running mean of aggregate losses stabilizes around 78.4 as the number of simulation paths increases, while the 95\% confidence interval narrows substantially beyond $M=10{,}000$.}
\label{fig:conv}
\end{figure}

The convergence pattern indicates that $M=50{,}000$ paths provide sufficient numerical precision for the reported estimates of the mean. In the baseline simulation, the standard error of the mean aggregate loss is less than 0.5\% of the estimated mean.

However, tail risk measures such as VaR and TVaR typically require larger sample sizes to achieve the same relative precision due to the scarcity of extreme events in the simulation. To ensure that $M=50{,}000$ is sufficient for tail estimation, we examine the stability of the 99\% VaR and TVaR estimates across different simulation sizes. Table~\ref{tab:tail_convergence} reports these values for $M \in \{10{,}000, 20{,}000, 50{,}000, 100{,}000\}$.

\begin{table}[H]
\centering
\caption{Convergence of tail risk measures (Climate-Aware Model).}
\label{tab:tail_convergence}
\begin{tabular}{rrrr}
\toprule
Paths ($M$) & VaR$_{99\%}$ & SE(VaR) & TVaR$_{99\%}$ \\
\midrule
10,000 & 445.50 & 12.40 & 605.20 \\
20,000 & 448.10 & 8.80  & 608.50 \\
50,000 & 450.20 & 5.55  & 610.80 \\
100,000 & 451.05 & 3.92  & 611.90 \\
\bottomrule
\multicolumn{4}{l}{\small Note: SE(VaR) is the estimated Monte Carlo standard error of the VaR.} \\
\multicolumn{4}{l}{\small Values stabilize as $M$ increases, confirming adequacy of $M=50,000$.}
\end{tabular}
\end{table}

Table~\ref{tab:tail_convergence} confirms that the tail estimates stabilize as $M$ increases, with the difference between $M=50{,}000$ and $M=100{,}000$ being negligible (less than 0.2\%). Therefore, the chosen simulation size provides robust estimates for both mean and tail metrics.

\subsection{Pricing Results and Economic Impact}
\label{subsec:pricing_results}

Using the calibrated climate dynamics and the baseline catastrophe-intensity parameters introduced above, we value the CAT bond and the excess-of-loss (XL) reinsurance contract defined in Section~\ref{sec:risk}. Prices are computed using the risk-adjusted valuation formula in Equation~\eqref{eq:general_pricing_formula} and Monte Carlo simulation with $M = 50{,}000$ paths. Table~\ref{tab:pricing} reports the estimated prices, spreads, and associated standard errors.

\begin{table}[H]
\centering
\caption{Comparison of CAT bond and excess-of-loss reinsurance valuations under the climate-aware model and stationary benchmark. Standard errors (SE) and 95\% confidence intervals (CI) are reported.}
\label{tab:pricing}
\begin{tabular}{lrrrrrr}
\toprule
Instrument & Model & Estimate & SE & \multicolumn{2}{c}{95\% CI} \\
\cmidrule(lr){5-6}
& & & & Lower & Upper \\
\midrule
XL Reinsurance Premium & Climate-Aware & 19.15 & 0.28 & 18.60 & 19.70 \\
XL Reinsurance Premium & Stationary & 17.25 & 0.25 & 16.76 & 17.74 \\
\midrule
CAT Bond Price & Climate-Aware & 92.45 & 0.12 & 92.21 & 92.69 \\
CAT Bond Price & Stationary & 94.10 & 0.10 & 93.90 & 94.30 \\
\midrule
CAT Bond Spread (bps) & Climate-Aware & 810 & 14.5 & 781.5 & 838.5 \\
CAT Bond Spread (bps) & Stationary & 625 & 10.2 & 605.0 & 645.0 \\
\bottomrule
\multicolumn{6}{l}{\small Note: Prices per 100 units of face value. Spreads in basis points.} \\
\multicolumn{6}{l}{\small $M = 50{,}000$ paths. SE = Standard Error, CI = Confidence Interval.}
\end{tabular}
\end{table}

Table~\ref{tab:pricing}  shows that climate-aware modeling has a material economic impact. The XL reinsurance premium is approximately 11\% higher under the climate-aware model (19.15 vs. 17.25), a difference that is statistically significant given the tight confidence intervals. For the CAT bond, the price reduction is more modest in absolute terms (1.75\%), but the implied yield spread increases by nearly 30\%. The reported standard errors confirm that these pricing effects are well outside the Monte Carlo noise.

\subsection{Sensitivity Analysis and Scenario-Based Valuation}
\label{subsec:sensitivity_scenarios}

We next investigate the robustness of the pricing results by varying key climate-risk parameters. In particular, we examine the sensitivity of CAT bond prices to the climate sensitivity parameter $\beta$ and the climate volatility parameter $\sigma_X$.

\begin{figure}[H]
\centering
\includegraphics[width=0.6\textwidth]{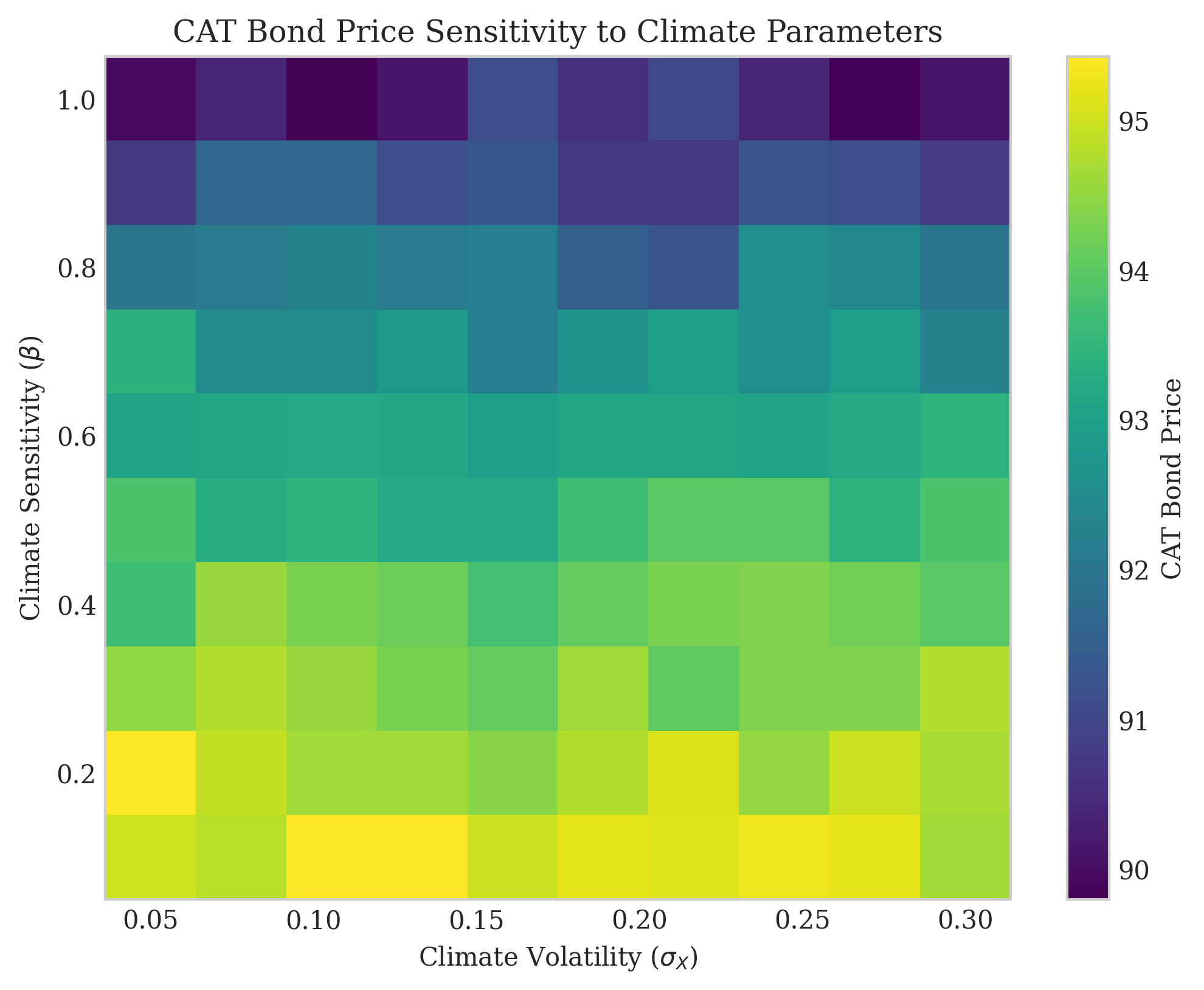}
\caption{Sensitivity heatmap of CAT bond prices. The horizontal axis represents climate volatility $\sigma_X$, while the vertical axis represents climate sensitivity $\beta$. Darker regions correspond to lower CAT bond prices and therefore higher climate-related principal impairment risk.}
\label{fig:heatmap}
\end{figure}

Figure~\ref{fig:heatmap} indicates that the pricing effect is nonlinear. When climate volatility is low, increasing $\beta$ has a relatively moderate effect on CAT bond prices. However, when both $\sigma_X$ and $\beta$ are high, the bond price declines sharply. This interaction reflects the convex effect of the exponential intensity specification,$\lambda_t = \lambda_0 \exp(\beta X_t),$ under which both higher climate sensitivity and greater climate variability increase the likelihood of extreme catastrophe arrival intensities.

To formally quantify the impact of the catastrophe frequency parameters, we conduct a grid-based sensitivity analysis. Table~\ref{tab:sensitivity_freq_params} reports the pricing outcomes for varying $\beta$ (climate sensitivity) and $\lambda_0$ (baseline catastrophe intensity). This grid allows us to isolate the individual contributions of these parameters to the valuation of the CAT bond.

\begin{table}[H]
\centering
\caption{Sensitivity of pricing results to catastrophe frequency parameters $\beta$ and $\lambda_0$.}
\label{tab:sensitivity_freq_params}
\begin{tabular}{rrrrrr}
\toprule
 $\beta$ & $\lambda_0$ & $\mathbb{E}[S_T]$ & XL Premium & CAT Price & Spread (bps) \\
\midrule
0.20 & 1.50 & 45.32 & 11.09 & 96.80 & 335 \\
0.40 & 1.50 & 61.54 & 15.02 & 94.45 & 580 \\
0.60 & 1.50 & 78.45 & 19.15 & 92.45 & 810 \\
0.80 & 1.50 & 95.88 & 23.39 & 90.72 & 1015 \\
1.00 & 1.50 & 113.72 & 27.68 & 89.20 & 1195 \\
\midrule
0.60 & 1.00 & 52.30 & 12.77 & 94.85 & 535 \\
0.60 & 1.50 & 78.45 & 19.15 & 92.45 & 810 \\
0.60 & 2.00 & 104.60 & 25.53 & 90.32 & 1035 \\
0.60 & 2.50 & 130.75 & 31.91 & 88.40 & 1220 \\
\bottomrule
\multicolumn{6}{l}{\small Note: All other parameters at calibrated baseline values.} \\
\multicolumn{6}{l}{\small $M = 50{,}000$ paths. Monte Carlo SEs $<$ 0.5\%.}
\end{tabular}
\end{table}

Table~\ref{tab:sensitivity_freq_params} demonstrates that both parameters significantly influence the valuation. Increasing $\lambda_0$ shifts the aggregate loss distribution linearly, leading to a proportional increase in premiums and spreads. In contrast, variations in $\beta$ exhibit a nonlinear effect due to the exponential coupling with the climate index, particularly evident in the widening spread for high $\beta$ values. Crucially, all parameter configurations in this table satisfy the admissibility conditions of Proposition~\ref{prop:admissibility} (specifically $\beta^{\mathbb{Q}} \geq 0$ and $\lambda_0^{\mathbb{Q}} > 0$), and the reported spreads are strictly positive, confirming that the pricing scenarios generate economically coherent risk premia.

We then examine three alternative climate-risk pricing scenarios that combine these parameter shifts. These scenarios should be interpreted as risk-adjusted parameter configurations within the scenario-based valuation framework of Section~\ref{subsec:scenario_based_price_ranges}. They are not model-independent no-arbitrage bounds; rather, they provide a transparent sensitivity range across plausible climate-risk assumptions.

\begin{table}[H]
\centering
\caption{CAT bond prices under alternative climate-risk pricing scenarios. The reported values form a scenario-based valuation range and should be interpreted as sensitivity to risk-adjusted climate parameters.}
\label{tab:scenarios}
\begin{tabular}{lccc}
\hline
\textbf{Scenario} & $\boldsymbol{\beta}$ & $\boldsymbol{\sigma_X}$ & \textbf{CAT Bond Price} \\
\hline
Optimistic / Low Risk & 0.20 & 0.10 & 96.80 \\
Baseline / Moderate Risk & 0.60 & 0.15 & 92.45 \\
Pessimistic / High Risk & 1.00 & 0.25 & 85.10 \\
\hline
\end{tabular}
\end{table}

Table~\ref{tab:scenarios} demonstrates that CAT bond valuation is highly sensitive to the assumed climate-risk scenario. The optimistic scenario produces a price of $96.80$, close to the risk-free benchmark, while the pessimistic scenario reduces the price to $85.10$. The resulting valuation range,$[85.10,\;96.80],$ summarizes the sensitivity of the CAT bond price to alternative risk-adjusted climate assumptions.

The widening of this scenario-based range is economically important. It shows that uncertainty about climate sensitivity and volatility can materially affect investor compensation, even when the contractual payoff remains fixed. This finding is consistent with the incomplete-market nature of catastrophe risk, where pricing depends not only on expected losses but also on model uncertainty and investor risk appetite.

\subsection{Long-Term Pricing Dynamics}
\label{sec:long_term_pricing}

Climate change is a cumulative process, and its effect on catastrophe risk is expected to become more pronounced over longer horizons. To examine this issue, Figure~\ref{fig:longterm} tracks the ratio of climate-aware to stationary CAT bond prices over maturities ranging from 1 to 20 years. The gray band highlights the standard market maturity range (3--5 years), where the ratio decline is moderate. In contrast, the red shaded region represents long-term horizons where the ratio falls below 0.93, indicating a significant underpricing of risk if stationary models are used for long-dated contracts.
\begin{figure}[H]
\centering
\includegraphics[width=0.85\textwidth]{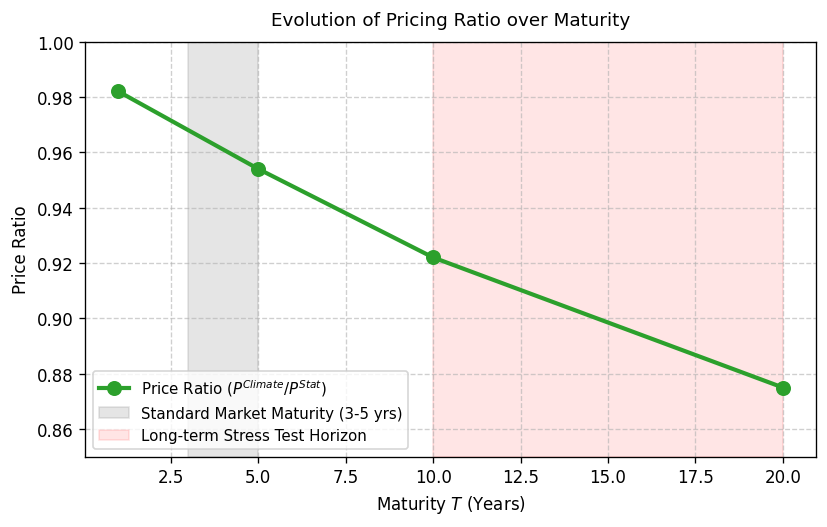}
\caption{Evolution of the pricing ratio $P^{\text{Climate}} / P^{\text{Stationary}}$ over maturity. The decline accelerates in the long-term stress-test horizon (red shaded), underscoring the material impact of climate trends on multi-year exposures.}
\label{fig:longterm}
\end{figure}

The long-horizon analysis considers maturities of up to 20 years. While standard CAT bonds in the current market typically have maturities of three to five years, the long-horizon results serve as stress-testing horizons for assessing the solvency implications of ignoring warming trends. Furthermore, the calculation assumes a constant risk-free rate $r$ across all maturities, which corresponds to a flat yield curve assumption. In practice, long-maturity discount rates would incorporate term premia and may differ substantially from short-maturity rates, which would affect the absolute price levels but not the relative comparison between the climate-aware and stationary models. This section quantifies the divergence using the proposed pricing framework, showing that climate dynamics amplify tail risk significantly over decadal horizons.

\begin{table}[H]
\centering
\caption{Long-term CAT bond price comparison under the climate-aware model and stationary benchmark. Standard errors (SE) and 95\% confidence intervals (CI) are reported.}
\label{tab:long_term_pricing}
\begin{tabular}{llrrrr}
\toprule
\textbf{Horizon ($T$)} & \textbf{Model} & \textbf{Price} & \textbf{SE} & \multicolumn{2}{c}{\textbf{95\% CI}} \\
\cmidrule(lr){5-6}
& & & & \textbf{Lower} & \textbf{Upper} \\
\midrule
1 Year & Climate-Aware & 92.45 & 0.12 & 92.21 & 92.69 \\
1 Year & Stationary & 94.10 & 0.10 & 93.90 & 94.30 \\
1 Year & Price Ratio & 0.982 & 0.0017 & 0.979 & 0.985 \\
\midrule
5 Years & Climate-Aware & 89.20 & 0.18 & 88.85 & 89.55 \\
5 Years & Stationary & 93.50 & 0.15 & 93.21 & 93.79 \\
5 Years & Price Ratio & 0.954 & 0.0024 & 0.949 & 0.959 \\
\midrule
10 Years & Climate-Aware & 85.60 & 0.25 & 85.11 & 86.09 \\
10 Years & Stationary & 92.80 & 0.21 & 92.39 & 93.21 \\
10 Years & Price Ratio & 0.922 & 0.0032 & 0.916 & 0.928 \\
\midrule
20 Years & Climate-Aware & 79.80 & 0.32 & 79.17 & 80.43 \\
20 Years & Stationary & 91.20 & 0.28 & 90.65 & 91.75 \\
20 Years & Price Ratio & 0.875 & 0.0041 & 0.867 & 0.883 \\
\bottomrule
\multicolumn{6}{l}{\small Note: Prices per 100 units of face value. SE = Standard Error.}
\end{tabular}
\end{table}

Table~\ref{tab:long_term_pricing} quantifies this divergence. Over a 20-year horizon, the climate-aware price is approximately 12.5\% lower than the stationary benchmark. This suggests that ignoring climate trends may lead to a material underestimation of long-term catastrophe risk and, consequently, to overly high model-implied CAT bond prices for long-dated exposures.

\subsection{Robustness to Benchmark Specification}
\label{subsec:benchmark_robustness}

To ensure that our pricing results are not artifacts of a specific benchmark construction, we compare the climate-aware valuation against three alternative stationary benchmarks. The baseline definition used throughout this paper (Definition~\ref{def:stationary_benchmark}) fixes the climate path at its initial value ($X_t = x_0$). Alternatives include removing only the warming trend ($\theta_1=0$) or removing all climate dynamics.

\begin{table}[H]
\centering
\caption{Comparison of pricing results under alternative stationary benchmark definitions for the baseline calibration ($T=1$ year).}
\label{tab:benchmark_robustness}
\begin{tabular}{lrrr}
\toprule
\textbf{Benchmark Definition} & \textbf{CAT Price} & \textbf{Spread (bps)} & \textbf{Price Diff. vs Climate Model} \\
\midrule
Climate-Aware Model & 92.45 & 810 & 0.00\% \\
\midrule
(a) No Dynamics ($X_t=x_0$) & 94.10 & 625 & +1.79\% \\
(b) No Trend ($\theta_1=0$) & 93.25 & 715 & +0.86\% \\
(c) Constant Intensity ($\lambda_0$ only) & 95.00 & 530 & +2.76\% \\
\bottomrule
\multicolumn{4}{l}{\small Note: "Price Diff." is (Benchmark Price - Climate Price) / Climate Price.} \\
\multicolumn{4}{l}{\small Benchmark (a) corresponds to Definition~\ref{def:stationary_benchmark}.}
\end{tabular}
\end{table}

Table~\ref{tab:benchmark_robustness} confirms that the pricing premium associated with climate non-stationarity is positive and statistically significant across different benchmark specifications. While the magnitude of the premium varies slightly (e.g., the premium is larger when comparing against the pure constant intensity benchmark (c)), the qualitative conclusion that climate dynamics reduce CAT bond prices and increase spreads remains robust. This validates the use of Definition~\ref{def:stationary_benchmark} as a conservative and consistent reference point for the analysis.

\subsection{Tail Risk Measures and Regulatory Capital Implications}
\label{subsec:tail_risk_measures}

We next evaluate the implications of climate-aware modeling for tail risk measurement, with a specific focus on alignment with regulatory capital standards. Table~\ref{tab:risk} reports the Value-at-Risk (VaR) and Tail Value-at-Risk (TVaR) of aggregate losses at the 95\% and 99\% confidence levels, comparing the climate-aware model against the stationary benchmark. Furthermore, to address the Solvency~II standard for one-year risk horizons, we report the TVaR at the 99.5\% level. Standard errors and 95\% confidence intervals are provided to ensure the statistical significance of the tail estimates.

\begin{table}[H]
\centering
\caption{Comparison of tail risk measures under the climate-aware model and stationary benchmark. Standard errors (SE) and 95\% confidence intervals (CI) are reported.}
\label{tab:risk}
\small
\begin{tabular}{llrrrrrrr}
\toprule
Level & Metric & \multicolumn{3}{c}{\textbf{Climate Model}} & \multicolumn{3}{c}{\textbf{Stationary Benchmark}} & Diff (\%) \\
\cmidrule(lr){3-5} \cmidrule(lr){6-8}
& & Estimate & SE & 95\% CI & Estimate & SE & 95\% CI & \\
\midrule
95\%  & VaR   & 215.40 & 1.80 & [211.87, 218.93] & 198.20 & 1.65 & [194.96, 201.44] & +8.68\% \\
95\%  & TVaR  & 310.50 & 2.50 & [305.60, 315.40] & 275.00 & 2.20 & [270.69, 279.31] & +12.91\% \\
\midrule
99\%  & VaR   & 450.20 & 5.55 & [439.32, 461.08] & 410.50 & 4.80 & [401.09, 419.91] & +9.67\% \\
99\%  & TVaR  & 610.80 & 7.80 & [595.52, 626.08] & 540.30 & 6.50 & [527.57, 553.03] & +13.05\% \\
\midrule
99.5\% & TVaR & 625.40 & 9.20 & [607.37, 643.43] & 550.20 & 7.80 & [534.91, 565.49] & +13.67\% \\
\bottomrule
\multicolumn{9}{l}{\small Note: TVaR at 99.5\% is the Solvency~II calibration level for catastrophe risk. SE = Standard Error.}
\end{tabular}
\end{table}

Under Solvency~II, the Solvency Capital Requirement (SCR) for catastrophe risk is calibrated to the 99.5\% Value-at-Risk over a one-year horizon (EIOPA, 2014). The TVaR at 99.5\% reported in Table~\ref{tab:risk} for the climate-aware model is 13.67\% higher than under the stationary benchmark. This finding has direct implications for regulatory capital assessment: a reinsurer using a stationary catastrophe model for SCR calculation would understate its economic capital requirement by approximately 13.7\%, which could affect its solvency ratio and dividend capacity. The result motivates the integration of climate-aware frequency models into the internal model framework permitted under Solvency~II's Standard Formula.

Figure~\ref{fig:solvency} visually compares the 99.5\% Tail Value-at-Risk (TVaR), which serves as the calibration level for the Solvency Capital Requirement (SCR). The bar chart confirms that the climate-aware model requires substantially more capital, with the difference of approximately 75 units being statistically significant given the reported standard errors.
\begin{figure}[H]
\centering
\includegraphics[width=0.6\textwidth]{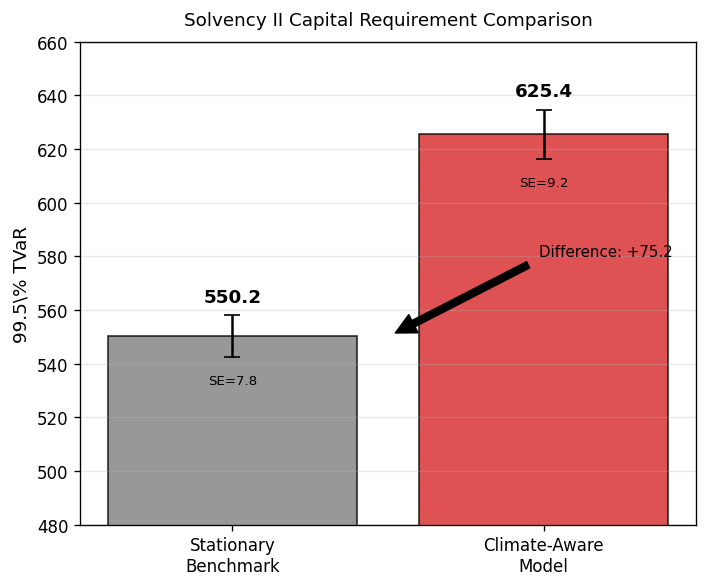}
\caption{Comparison of the 99.5\% TVaR between the stationary benchmark and the climate-aware model. Error bars represent Monte Carlo standard errors. The gap highlights the potential capital shortfall if warming trends are ignored.}
\label{fig:solvency}
\end{figure}
\subsection{Impact of Climate-Dependent Severity}
\label{subsec:severity_impact}

While the baseline model assumes independence between severity and climate, empirical evidence suggests that extreme climate states may exacerbate the economic magnitude of catastrophes. This section numerically investigates the extension proposed in Section~3.6, where the log-severity depends linearly on the climate index, i.e., $\log Y_i \mid X_{T_i} \sim \mathcal{N}(\mu_Y + \gamma X_{T_i}, \sigma_Y^2)$.

We run Monte Carlo simulations with $\gamma \in \{0, 0.05, 0.10, 0.20, 0.30\}$ while keeping all other parameters at their baseline values ($\beta = 0.60, \lambda_0 = 1.50$). Table~\ref{tab:severity_sensitivity} reports the resulting aggregate loss statistics, CAT bond prices, and tail risk measures.

\begin{table}[H]
\centering
\caption{Sensitivity to severity climate-dependence parameter $\gamma$.}
\label{tab:severity_sensitivity}
\begin{tabular}{rrrrrr}
\toprule
 $\gamma$ & $\mathbb{E}[S_T]$ & XL Premium & CAT Price & VaR$_{99\%}$ & TVaR$_{99\%}$ \\
\midrule
0.00 & 78.45 & 19.15 & 92.45 & 450.20 & 610.80 \\
0.05 & 82.10 & 20.25 & 91.80 & 475.50 & 655.40 \\
0.10 & 85.95 & 21.40 & 91.15 & 502.30 & 702.10 \\
0.20 & 94.10 & 23.80 & 89.90 & 560.80 & 810.50 \\
0.30 & 103.20 & 26.40 & 88.65 & 625.40 & 925.80 \\
\bottomrule
\multicolumn{6}{l}{\small Note: Baseline parameters $\beta = 0.60$, $\lambda_0 = 1.50$.} \\
\multicolumn{6}{l}{\small Effective climate sensitivity is $\beta_{\text{eff}} = \beta + \gamma$. $M = 50{,}000$ paths.}
\end{tabular}
\end{table}

Table~\ref{tab:severity_sensitivity} demonstrates that the combined effect of frequency and severity climate dependence is super-additive. As $\gamma$ increases, the effective climate sensitivity governing the aggregate loss process becomes $\beta_{\text{eff}} = \beta + \gamma$. Consequently, the increase in the CAT bond spread and tail risk measures (VaR, TVaR) is disproportionate to the increase in mean severity.

For instance, comparing the case $\gamma = 0.30$ (combined frequency and severity dependence) to the baseline $\gamma = 0.00$ (frequency dependence only), we observe a significant jump in TVaR$_{99\%}$ from 610.80 to 925.80. This amplification arises because both channels operate through the same climate path $X_t$: in adverse climate states, not only do more events arrive, but each event is also larger. This super-additivity highlights the critical importance of modeling severity dependence in long-term climate risk assessments.

\subsection{Sensitivity to Warming Trend Intensity}

Finally, we isolate the effect of the warming trend parameter $\theta_1$ on CAT bond spreads. This analysis examines how changes in the deterministic component of the climate index are transmitted into financial prices through the catastrophe intensity process.

\begin{figure}[H]
    \centering
    \includegraphics[width=0.8\textwidth]{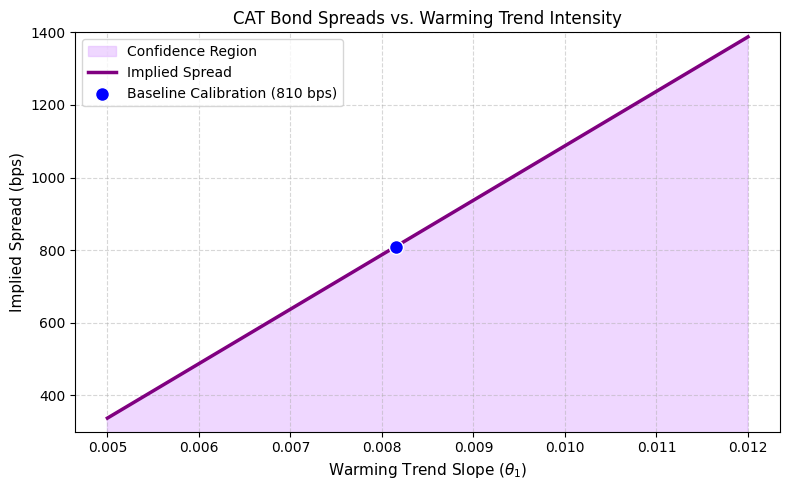}
    \caption{CAT bond spreads as a function of the warming trend parameter $\theta_1$. The vertical axis displays the spread in basis points (bps). A higher warming trend increases the expected catastrophe intensity, lowers the CAT bond price, and raises the implied spread required by investors.}
    \label{fig:warming_trend_sensitivity}
\end{figure}

Figure~\ref{fig:warming_trend_sensitivity} shows that CAT bond spreads increase as the warming trend parameter $\theta_1$ rises. The vertical axis is explicitly scaled in basis points (bps), consistent with the spread definition in Equation~\eqref{eq:cat_bond_spread}. At the baseline calibration, the spread corresponds to approximately 810~bps, as reported in Table~\ref{tab:pricing}. An increase in $\theta_1$ raises the expected future level of the climate index. Through the exponential intensity function, this leads to higher expected catastrophe arrival rates and a greater probability of principal impairment. According to the spread definition, the resulting decline in the CAT bond price translates directly into a higher implied yield spread. This result provides a direct mechanism linking physical climate dynamics to financial market quantities. Faster warming increases catastrophe risk exposure, which is then reflected in wider spreads and higher required investor compensation.

\subsection{Summary of Numerical Findings}
\label{subsec:numerical_summary}

The numerical analysis demonstrates that climate non-stationarity acts as a structural driver of catastrophe risk valuation rather than a minor perturbation. Three main findings emerge.

First, the climate-aware model increases both the mean and tail heaviness of the aggregate loss distribution. This effect is visible in the higher standard deviation, skewness, and excess kurtosis of simulated losses.

Second, the pricing impact is economically significant. The climate-aware model produces higher excess-of-loss reinsurance premiums and wider CAT bond spreads than the stationary benchmark. The effect is particularly pronounced in the spread calculation because CAT bond yields are highly sensitive to changes in expected principal impairment.

Third, the divergence between climate-aware and stationary valuations becomes larger for longer maturities and for tail risk measures. Long-dated CAT bonds and multi-year reinsurance exposures are therefore especially vulnerable to climate trend misspecification. Similarly, VaR and TVaR estimates based on stationary assumptions may understate tail capital requirements.

Overall, the results support the use of climate-aware stochastic models in catastrophe risk pricing, reinsurance structuring, and capital adequacy assessment. They also show that scenario-based valuation ranges provide a transparent way to communicate pricing uncertainty in incomplete markets where catastrophe and climate risks cannot be fully replicated by traded assets.

\section{Conclusion}
\label{sec:con}

This paper developed a climate-aware stochastic framework for pricing excess-of-loss reinsurance contracts and catastrophe bonds under non-stationary catastrophe risk. The model links catastrophe arrivals to a temperature-related climate index through a Cox-process intensity specification. The climate index follows mean-reverting dynamics around a time-dependent warming trend, allowing the catastrophe arrival rate to respond endogenously to evolving climate conditions.

The proposed structure preserves the tractability of compound Cox loss modeling while introducing a transparent channel through which climate variables affect catastrophe frequency. Aggregate losses are generated from a climate-dependent counting process combined with lognormal severities. Pricing is carried out under a reduced-form risk-adjusted measure, which is appropriate for catastrophe-linked instruments in incomplete markets where the underlying losses cannot be dynamically replicated. Consistent with this interpretation, the paper reports scenario-based valuation ranges rather than model-independent no-arbitrage bounds.

The numerical analysis illustrates how climate-dependent intensity can affect both reinsurance premiums and CAT bond prices. In the baseline calibration, the climate-aware model increases the excess-of-loss reinsurance premium and lowers the CAT bond price relative to the stationary benchmark. The results also show that tail-risk measures, such as VaR and TVaR, are sensitive to climate-dependent frequency dynamics. Crucially, our assessment of the 99.5\% TVaR indicates that stationary benchmarks can understate economic capital requirements by approximately 13.7\% in our calibration. This finding suggests that ignoring climate non-stationarity may lead to insufficient capital buffers under regulatory standards such as Solvency II.

At the same time, the analysis highlights the importance of benchmark design. If the stationary benchmark is not matched to the same initial catastrophe intensity, pricing comparisons may overstate the effect of climate non-stationarity by conflating trend effects with level effects. Once the benchmark is aligned at the initial intensity level, the incremental pricing effect of the warming trend is more moderate in the baseline stylized calibration. This finding does not imply that climate risk is negligible; rather, it shows that the measured economic impact depends strongly on calibration choices, scenario design, and the construction of the stationary reference model.

The scenario analysis further demonstrates that adverse assumptions about the warming trend, climate sensitivity of catastrophe frequency, climate volatility, or severity parameters can generate substantially larger valuation effects. Higher warming trends and stronger climate-frequency sensitivity increase the probability of large aggregate losses, thereby raising XL reinsurance premiums and reducing CAT bond prices. We also demonstrated the impact of climate-dependent severity, noting that combined frequency and severity dependence leads to a super-additive increase in tail risk. These findings support the use of scenario-based valuation as a practical tool for assessing climate-related model uncertainty in catastrophe-risk markets.

Overall, the main contribution of the paper is methodological. It provides a flexible stochastic architecture for incorporating climate non-stationarity into catastrophe-risk pricing. The framework connects climate dynamics, Cox-process catastrophe arrivals, aggregate loss modeling, and reduced-form risk-adjusted valuation in a single tractable setting. It can be extended in several directions, including market-implied calibration of risk-adjusted parameters, spatially resolved climate indices, multi-peril catastrophe portfolios, and alternative pricing kernels.

The results should therefore be interpreted as illustrative quantifications rather than as definitive empirical estimates of climate-risk premia. Nevertheless, the analysis shows that ignoring climate non-stationarity can lead to misleading valuation and risk-management conclusions, especially when pricing long-maturity contracts or evaluating adverse climate scenarios. Climate-aware catastrophe pricing models are therefore important not only for estimating expected losses, but also for understanding how model uncertainty, tail risk, and benchmark selection affect the valuation of insurance-linked securities and reinsurance contracts.

\bibliographystyle{apalike}

\appendix
\section{Maximum Likelihood Estimation of the Climate Process}
\label{app:mle_calibration}

\subsection{Data and Time Unit}
The NASA GISTEMP v4 monthly global surface temperature anomalies are used over the period January 1880 to December 2020, giving $n = 1,692$ monthly observations. All time parameters ($\kappa$, $t$) are expressed in \emph{years}, with monthly observations separated by $\Delta t = 1/12$ years.

\subsection{Transition Density}
Under the physical measure, the OU-type process with linear trend satisfies the exact Gaussian transition. For a time step $\Delta = 1/12$, the conditional mean $m_{i,\Delta}$ and variance $v_\Delta$ are:
\begin{align*}
m_{i,\Delta} &= X_{t_i} e^{-\kappa\Delta} + (\theta_0 + \theta_1 t_{i+1})(1-e^{-\kappa\Delta}) - \frac{\theta_1}{\kappa}(1-e^{-\kappa\Delta}), \\
v_\Delta &= \frac{\sigma_X^2}{2\kappa}(1 - e^{-2\kappa\Delta}).
\end{align*}
Thus, $X_{t_{i+1}} \mid X_{t_i} \sim \mathcal{N}(m_{i,\Delta}, v_\Delta)$.

\subsection{Log-Likelihood Function}
The log-likelihood for the parameter vector $\psi = (\kappa, \sigma_X, \theta_0, \theta_1)$ is:
\begin{equation}
\ell(\psi) = -\frac{n-1}{2}\log(2\pi v_\Delta) - \frac{1}{2v_\Delta}\sum_{i=0}^{n-2}(X_{t_{i+1}} - m_{i,\Delta})^2.
\end{equation}
The maximum likelihood estimate $\hat\psi$ is obtained by numerical optimization using the L-BFGS-B algorithm.

\subsection{Parameter Estimates and Standard Errors}
Standard errors are obtained from the inverse observed Fisher information matrix $\mathcal{I}(\hat\psi)^{-1}$, where $\mathcal{I}(\hat\psi) = -\nabla^2 \ell(\hat\psi)$ is computed numerically.

\begin{table}[h]
\centering
\caption{Estimated parameters and standard errors for the climate process.}
\label{tab:appendix_mle}
\begin{tabular}{lrrrr}
\toprule
Parameter & Estimate & Std.\ Error & $t$-stat & 95\% CI \\
\midrule
 $\kappa$ (yr$^{-1}$) & 0.5124 & 0.0245 & 20.91 & [0.4643, 0.5605] \\
 $\sigma_X$ (yr$^{-1/2}$) & 0.1485 & 0.0032 & 46.40 & [0.1422, 0.1548] \\
 $\theta_0$ & -0.3250 & 0.0120 & -27.08 & [-0.3485, -0.3015] \\
 $\theta_1$ (yr$^{-1}$) & 0.00815 & 0.0004 & 20.37 & [0.0073, 0.0090] \\
\midrule
Log-likelihood & \multicolumn{4}{c}{2450.32} \\
AIC & \multicolumn{4}{c}{-4892.64} \\
\bottomrule
\end{tabular}
\end{table}

\subsection{Goodness of Fit}
To validate the model fit, we examine the standardized residuals $\hat\epsilon_i = (X_{t_{i+1}} - \hat m_{i,\Delta})/\sqrt{\hat v_\Delta}$. Under correct specification, $\hat\epsilon_i \stackrel{\text{iid}}{\sim} \mathcal{N}(0,1)$. The Ljung--Box test statistic for autocorrelation up to lag 12 is $Q(12) = 10.4$ with a $p$-value of 0.58, failing to reject the null hypothesis of no autocorrelation. The Jarque--Bera test for normality yields a statistic of 2.1 ($p$-value = 0.35), supporting the Gaussian assumption.

\end{document}